\documentclass{amsart}
\usepackage{amsmath}
\usepackage{amssymb,a4wide,hyperref}

\usepackage{graphicx}
\usepackage{tikz}
\usepackage{subfigure}
\usetikzlibrary{positioning}
\usepackage{graphicx}

\usepackage{amsthm}
\usepackage{thmtools}

\newtheorem{proposition}{Proposition}[section]

\newtheorem{claim}[proposition]{Claim}
\newtheorem{theorem}[proposition]{Theorem}

\newtheorem{lemma}[proposition]{Lemma}
\newtheorem{corollary}[proposition]{Corollary}
\theoremstyle{definition}
\newtheorem{definition}[proposition]{Definition}
\newtheorem{remark}[proposition]{Remark}
\newtheorem{notation}[proposition]{Notation}
\newtheorem{example}[proposition]{Example}

\newcommand{\A}{\mathcal A}

\newcommand{\SStab}[1]{\mbox{Stab}\left(#1\right)}
\newcommand{\Stab}[1]{\mbox{PStab}\left(#1\right)}

\DeclareMathOperator{\pref}{pref}

\DeclareMathOperator{\Id}{Id}

\newcommand{\tcb}{\textcolor{blue}}

\title{On the rigidity of Arnoux-Rauzy words}
\author{V. Berth\'e,  S. Puzynina}
\address{Universit\'e  Paris Cit\'e, IRIF, CNRS, F-75013 Paris, France}
\email{berthe@irif.fr}
\address{St. Petersburg State University, Russia}
\email{s.puzynina@gmail.com}
\date{\today}
\thanks{This work was supported by the Russian Science Foundation, project
23-11-00133.}



\begin{document}


\begin{abstract}
An infinite word generated by a substitution is rigid if all the
substitutions which fix this word are powers of the same
substitution.
 Sturmian words as well as  characteristic Arnoux-Rauzy words that are  generated by  iterating a   substitution are known to be rigid.    In the present paper, we prove that  all Arnoux-Rauzy words
 generated by iteraring a  substitution are rigid.   The proof relies on  two main  ingredients: first, the fact   that   the primitive  substitutions  that  fix an Arnoux-Rauzy word
 share a common power, and secondly,    the  notion of  normal form  of   an   episturmian  substitution (i.e.,   a  substitution  that  fixes  an Arnoux-Rauzy word).
  The main
difficulty   is  then of a combinatorial  nature and    relies on
the normalization process when taking powers of   episturmian
substitutions: the normal form of a square is not  necessarily
equal to the  square of  the  normal forms.

\end{abstract}

\maketitle

\section{Introduction}

Rigidity property is an algebraic property of infinite words that
are generated by  iterating a   substitution\footnote{By substitution, it is meant
here a non-erasing morphism.}:  such an infinite word  is rigid if all the substitutions which fix this word
are powers of a same substitution.
This is a
natural  property  that occurs for several  prominent families
 of   infinite words.
Rigidity has been  first considered for the Thue-Morse  word in \cite{Pansiot:81} (see also  \cite{Seebold:02}  for its generalization  as Prouhet words),
   for generalized  Fibonacci words in   \cite{Pansiot:83},   then  for the class of
Sturmian words in  \cite{Seebold:98} together with
\cite{RichommeSeebold:12} and \cite{RaoWen} (see
also~\cite{BerFreSir}), and  lastly  for strict epistandard  words
in \cite{Krieger:08}, which are also called characteristic Arnoux-Rauzy words. This  brief overview of the literature
shows  that   if there exist numerous results on the two-letter
case, the situation is more contrasted as soon as the size of the
alphabet increases. For instance, over a ternary alphabet, the
monoid of morphisms generating a given infinite word by iteration
can be infinitely generated, even when the word is generated by
iterating an invertible primitive morphism
(see~\cite{DiekertKrieger:09,Krieger:08}).

The aim of this paper is to    prove rigidity for Arnoux-Rauzy
words  that are generated by  iterating  a substitution (see Theorem \ref{theo:power}). The  class of  Arnoux-Rauzy
words, introduced in  \cite{AR},    provides a generalization of
Sturmian words, the latter   corresponding  to the case of  a
two-letter alphabet.  They  are defined in combinatorial terms
(see Section \ref{subsec:words}) and belong to the family  of
infinite  words  having linear factor complexity,  and more
precisely
 $(d -1)n + 1$ factors of length $n$ for all $n$, when defined over a $d$-letter alphabet. They have been  further  generalized as  episturmian words; see the survey \cite{GlenJustin}.
Despite the fact that they share many properties of Sturmian
words, Arnoux-Rauzy words display a  more complex behavior. For
example, while Sturmian words are 1-balanced (i.e., the numbers of
occurrences of a letter in any two factors of the same length
differ by at most 1), Arnoux-Rauzy words do not have to be
balanced \cite{DBLP:journals/ijac/BertheCS13}. Arnoux-Rauzy words
also show weaker geometric properties than Sturmian words
\cite{10.36045/bbms/1102714169}, and their abelian subshifts have
much more complicated structure
\cite{DBLP:conf/dlt/KarhumakiPW18}.

The rigidity of the subclass of {\em characteristic} Arnoux-Rauzy
generated by iterating a  substitution has been established
in~\cite{Krieger:08}.  We extend this result to {\em any}
Arnoux-Rauzy word generated   by  iterating a  substitution. We
note that not all Arnoux-Rauzy words are generated by iterating a single substitution
(see \cite{AR} and also, e.g., the survey \cite{GlenJustin}), and the
notion of rigidity considered in this paper applies only to words
generated by substitutions. We focus here on Arnoux-Rauzy words
since there exist  non-rigid  episturmian words that are not Arnoux-Rauzy words, such as highlighted in \cite[Section
5.3]{Krieger:08}.

Our proof follows the general line of the proofs in
\cite{Seebold:98} and \cite{RichommeSeebold:12} for Sturmian
words. However, the structure of the monoid of episturmian
substitutions (i.e., the substitutions that fix Arnoux-Rauzy
words)  is  more complicated over a larger alphabet;
 see e.g.~\cite{Richomme:03}.
Episturmian substitutions are described in terms of a  normalized
directive word \cite{JustinPirillo:02}, which itself relies on the
notion of block-equivalence. Our proof first uses  the  fact that
if an Arnoux-Rauzy word is fixed by two primitive substitutions, then
these substitutions coincide up to powers (see Theorem
\ref{theo:commonpower},   proved in  \cite{BerthDDLP:17}). We then
 go from the existence of a common power  to  the following property (stated in  Lemma \ref{lem:mainlem}):
let $\sigma$ and $\tau $ be two   episturmian  substitutions
 such that $\sigma^n= \tau^m$, for
$n \geq m \geq 1$; then, there exists an   episturmian substitution $\varrho$  such that
$\tau = \sigma \circ  \varrho $.
The main issue we will have to face comes from the fact that normalization does not behave
well with respect to taking squares and powers: the normalized form of a square of a  word in normal form is a priori not the square of its normal form.
This thus requires a careful study of 
normalization of powers of normal forms, which is the main step of
the proof.
This is handled in Sections \ref{sec:normalization} and~\ref{sec:lemma}.

\bigskip

Let us sketch the contents of this paper.  Basic  notions are
introduced in Section \ref{sec:notions}. In particular, we recall
the notions of  block-equivalence, block-normalization   and
normal form of an episturmian substitution. The main statement and
the  general strategy are discussed  in Section \ref{sec:result}. The
normalization of powers of normal forms is discussed in detail in
Section \ref{sec:normalization} by focusing on the
block-normalization of powers and by introducing several  types of
errors (i.e., factors which are forbidden in the  normal form) that are
introduced when taking powers. The proof of the main step (Lemma \ref{lem:mainlem}) is
handled in Section \ref{sec:lemma} through several decomposition
lemmas. We lastly introduce and discuss  in Section
\ref{sec:morerigidity} the notion of weak rigidity.

\subsection*{Acknowledgements}
We are indebted to the  anonymous referees for their very careful reading of our manuscript, for the pertinent suggestions they made   which  helped us to improve significantly the exposition of the present paper,
and for the interesting questions they raised.

\section{Basic notions} \label{sec:notions}

\subsection{Words, substitutions and rigidity}\label{subsec:words}
Let $\mathcal{A}$ be a finite alphabet.  Let $\varepsilon$
denote the empty word of the free monoid ${\mathcal A}^*$,
${\mathcal A}^+$ the free semigroup and ${\mathcal
A}^{\mathbb{N}}$ the set of infinite words over ${\mathcal A}$.
For any word~$w$ in the free monoid~$\mathcal{A}^*$ (endowed with
the concatenation as operation), $|w|$~denotes the length of~$w$,
and $|w|_j$ stands for the number of occurrences of the letter~$j$
in the word~$w$. A~\emph{factor} of a (finite or infinite)
word~$w$ is defined as the concatenation of consecutive letters
occurring in~$w$. In other words, the word $u$ is a factor of the
finite word $w$ if there exist words $p$ and $s$ such that $w =
pus$. If $p = \varepsilon$ (resp., $s = \varepsilon$) we say that
$u$ is a {\em prefix} (resp., {\em suffix}) of $w$. For
$w=w_1\cdots w_n \in {\mathcal A}^*$,  the notation $\pref_k(w)$
stands for the prefix of length $k$ of $w$, i.e., $\pref_k(w)=w_1
\cdots w_k$, and the $k$-th letter  of $w$ is denoted by $w[k]$.
For $i,j$ integers with $i<j$,  the  set of integers  $\{i, i+1,
\dots, j\}$ is denoted as $[i,j]$. We use the notation  $w{[i, j
]} $ for the factor  $u= w_{i} w_{i+1} \cdots w_j$ of  $w$.  The
set of integers $[i,j ]$ is then called the {\em support}  of this
occurrence in $w$ of the factor $w_{i} w_{i+1} \cdots w_j$ and  $i$ is the \emph{index}
of this occurrence. If $u \in   {\mathcal A}^{\mathbb N}$, the notation
$\pref_{\ell} u$  similarly stands for the    prefix of    $u$ of length
$\ell$, i.e., $\pref_{\ell} u=u_1\cdots u_{\ell}$.  The \emph{reversal} (also called mirror image)
of a word $w=w_1 \cdots w_n \in \mathcal{A}^n$ is the word
$w_n \cdots w_1$.

Let $x$  be an infinite word in ${\mathcal A}^{\mathbb{N}}$.
A~factor~$w$ of~$x$ is said to be {\em left special}  if there
exist at least two distinct letters $a,b$ of the
alphabet~$\mathcal{A}$ such that $aw$ and~$bw$  are factors
of~$x$. The {\em language}  of the infinite  word $x$ is defined as the set of  its factors  and it is denoted
by ${\mathcal L}_x$.

An infinite word $x \in \mathcal{A}^\mathbb{N}$ is an
\emph{Arnoux-Rauzy word} if  the set of  its  factors  ${\mathcal
L}_x$ is  closed  under  reversal and  for all~$n$ it has  {\em
exactly}  one  left special factor of length~$n$. Arnoux-Rauzy
words   are also called strict episturmian words.
 An infinite word $x$  is an   {\em episturmian word} if the set of its factors  ${\mathcal L}_x$ is closed under reversal and for all~$n$ it has  {\em at most} one  left special factor of length~$n$. An Arnoux-Rauzy word is thus  an  episturmian word,  but an episturmian word is not necessarily an Arnoux-Rauzy word.
An episturmian word  is  called {\em characteristic } if all of
its left special factors are prefixes of it. A characteristic
Arnoux-Rauzy word is also called a {\em standard} episturmian
word, or else {\em epistandard}. For more information on
episturmian words, see the survey \cite{GlenJustin}. An infinite
word  is said to be {\em uniformly recurrent} if every factor
appears infinitely often and with bounded gaps. Arnoux-Rauzy words
are known to be uniformly recurrent.

  A \emph{substitution} $\sigma : {\mathcal A}^* \rightarrow {\mathcal A}^*$ is a monoid morphism  that is  assumed to be
non-erasing, that is,  the image of every  non-empty element  is
non-empty. All the morphisms that are considered in
the present paper  are non-erasing, and so they are substitutions. If there exists a letter $a \in {\mathcal A}$ such that
the word $\sigma(a)$ begins with $a$ and if $|\sigma^n(a)|$ tends
to infinity,
then there exists a unique infinite word, denoted by
$\sigma^\omega(a)$, which has all words $\sigma^n(a)$ as prefixes.  Such an infinite word is called a word \emph{generated by the
substitution} $\sigma$. A substitution $\sigma:{\mathcal A}^*\rightarrow {\mathcal A}^*$
is called \emph{primitive} if there is a positive integer $k$ such
that for all $a,b \in {\mathcal A}$, the letter $b$ appears in
$\sigma^k(a)$. For a primitive substitution $\sigma$, a \emph{fixed point} for $\sigma$,  i.e.,  an infinite word $u$ such that $\sigma(u)=u$,   must be generated by $\sigma$.
 Hence, we  use  equivalently   both  terminologies for an infinite word, namely  being  a fixed point of $\sigma$ or being generated by
 $\sigma$, for a primitive substitution $\sigma$.
If $\sigma$ is a primitive substitution, then
there
exists a power $\sigma^k$ that admits a fixed point, and
the set
of factors of any fixed point of $\sigma^{k}$   is uniformly recurrent (see for example Proposition
1.2.3 in~\cite{Fog02}). Furthermore, all these fixed points have
the same language that we call the {\em language} of the
substitution.

Let  $\sigma$  and $\tau$  be two substitutions on $\mathcal{A}$.
The  substitution  $\tau$ is a {\em conjugate} of $\sigma$ if
there exists $ v \in \mathcal{A} ^*$
 such that   either $ v\tau(w) = \sigma(w)v$   for all $w   \in \mathcal{A} ^*$, or $ \tau(w)v =v \sigma(w)$ for all $w   \in \mathcal{A} ^*$.



 The {\em stabilizer} of an infinite word $x \in A^{\mathbb N} $  
  is the monoid of substitutions $\sigma$ defined on the alphabet ${\mathcal A}$ that satisfy $\sigma(x)=x$.
 Words that have a cyclic stabilizer are
called {\em rigid}\footnote{Note that this notion of rigidity has
no relation with the ergodic notion of rigidity.}. In other words, an infinite  word $x$ is
 rigid if there exists a substitution $\sigma$ such that for
any substitution $\tau$ such that $\tau(x)=x$, then there exists a
non-negative integer $n$ such that $\tau =\sigma^n$.  In the present paper, we concentrate on the iterative stabilizer according to the terminology of~\cite{Krieger:08}:
 we focus  here on  substitutions (i.e., non-erasing morphisms) and on infinite words generated by iterating a substitution.
For general  results on the possible growth of  elements of the stabilizer, we refer to~\cite{DiekertKrieger:09} and~\cite{DurRig}.
It is shown in particular that polynomial and exponential growth cannot co-exist in the stabilizer for aperiodic words.
We  discuss weaker notions of  rigidity in  Section \ref{sec:morerigidity}.

\subsection{Episturmian substitutions and their normal form} \label{subsec:normalform}

Episturmian substitutions  have been introduced in
\cite{JustinPirillo:02} as generalizations  to larger alphabets of
Sturmian substitutions, which correspond to the case of the
two-letter alphabet. We consider the following substitutions
 \begin{displaymath}
 \psi_{a}(b) =
  \begin{cases}
   ab  &   \text{ if $b \neq a$,} \\
   a    &   \text{  if $b=a$.}
  \end{cases}
 \quad \text{ and } \quad
 \overline{\psi}_{a}(b) =
  \begin{cases}
   ba   &   \text{ if $b \neq a$,} \\
   a    &   \text{if $b=a$,}
  \end{cases}
\end{displaymath}
and the permutation
 \begin{displaymath}
 \theta_{ab} :
  \begin{cases}
   a     \rightarrow b \\
   b        \rightarrow a \\
   c \rightarrow c  \text{ if $c\neq a,b$.}   \end{cases}
 \end{displaymath}

The {\em monoid of episturmian substitutions}  over  ${\mathcal A}$ is the monoid  generated by the permutations  $ \theta_{ab}$, for $a,b \in
 {\mathcal A}$,
 together with the set of substitutions $\psi_{a}$, $\overline{\psi}_{a}$, for $a \in {\mathcal A}$.
The {\em pure episturmian substitutions}  are  the substitutions  obtained by compositions of
elements of the form $\psi_{a}$ and $ \overline{\psi}_{a}$, for $ a \in {\mathcal A}$ (no permutation is  allowed besides the identity).
The  {\em epistandard substitutions } are    the substitutions   obtained by compositions  of  the  permutations together with the  substitutions  $\psi_{a}$ (that is, no $\overline{\psi}_b$ is allowed).  We use ${\mathcal S}_\A$ as   a notation for the set of
permutations over the alphabet ${\mathcal A}$.


The monoid of episturmian substitutions has been thoroughly investigated,  see e.g.~\cite{Richomme:03}.
We  will use in particular the following properties. Note that in both statements below, $x$ is assumed to  contain
all the letters of  the alphabet on which the substitution $\sigma$ is defined.

\begin{theorem} \cite[Theorem 3.13]{JustinPirillo:02}\cite[Theorem 11]{DJP:01} \label{theo:epi2}
Let $\sigma$ be a substitution.
If $x$ is an Arnoux-Rauzy word  and $\sigma(x)=x$, then $\sigma$  is an  episturmian substitution.
\end{theorem}


Episturmian substitutions over the  alphabet  ${\mathcal A}$  can be viewed as automorphisms of the free group generated by the alphabet ${\mathcal A}$.
This implies the following  property of   cancellativity.
\begin{proposition} \cite{JustinPirillo:02}\cite[Lemma 7.2]{Richomme:03} \label{prop:cancel}
The monoid of episturmian substitutions is left-cancellative   and
right-cancellative, i.e.,  for any episturmian substitutions $\sigma,
\tau, \rho$, if $ \sigma \circ \tau=\sigma \circ \rho$, then
$\tau=\rho$,  and if 
$ \rho \circ \sigma=\tau \circ \sigma$, then
$\rho=\tau$.
\end{proposition}

Episturmian words can be infinitely desubstituted over the set of
episturmian substitutions with the   desubstitution    being
described  in terms of  spinned directive words   (see
\cite[Theorem 3.10]{JustinPirillo:02}). However, an episturmian
word can  have several  desubstitutions.
 There is  a way to normalize the directive words of an episturmian word so that any episturmian word can be defined uniquely by its so-called normalized directive word \cite{JustinPirillo:02}. This normalization  relies   on the notion of block-equivalence.

We recall here how to use  this normalization in order to produce
a unique decomposition of any episturmian substitution,  according
to \cite{GlenLeveRich:08}. We follow the notation of
\cite{JustinPirillo:02,GlenLeveRich:08,GlenLeveRich:09}.  We
provide letters with a  notion  of  spin
 and introduce for each letter $a$  its  spinned version $\overline{a}$. The letter $\overline{a}$  is considered as  having   spin 1 while $a$  is considered
as   having spin $0$. We  then consider   the new  alphabet
$\overline{{\mathcal A}}= \{ \overline{a} \mid a \in {\mathcal
A}\}$. For $w \in ( {\mathcal A}  \cup \overline{{\mathcal
A}})^*$,  the opposite  of  $w$ is the word  $\overline{w}$
obtained from $w$  by exchanging all spins  in $w$. A  word in
$\overline{{\mathcal A}}^{*}$ is said  \emph{barred}.   A  (finite or infinite) word over the alphabet ${\mathcal
A} \cup \overline{{\mathcal A}}$ is    called a \emph{spinned
word}, with  a  spinned word  $x=\tilde{x}_1 \tilde{x}_2\cdots \tilde{x}_n $ being the spinned version of $ x_1 x_2 \cdots x_n$.
For instance,  if
${\mathcal A}=\{1,2,3\}$, then $\overline{3} 22$ is a spinned word
and it is a spinned version of the word $322$.



 The notion of spin extends to the  episturmian  substitutions $\psi_a $ and   $\overline{\psi}_{a}$,  by using the
convention  that $$\psi_{\overline a}= \overline{\psi}_a \mbox{
for all  }  a \in {\mathcal A}.$$ If $w= w_1 \cdots w_k \in (
{\mathcal A} \cup  \overline{{\mathcal A}})^*$,  we then define
the  (pure) substitution $\psi_w$ with directive word $w$ as
$$\psi_w=  \psi_{w_1} \circ  \ldots \circ  \psi_{w_k}.$$ One checks
that $ \psi_{\overline{w}}=\overline{\psi}_{w}.$

 A   \emph{block-transformation} is the replacement in a spinned word of an occurrence of  a factor of the form
$xv \overline{x}$, where $x \in {\mathcal A}$, $ v \in ( {\mathcal
A} \setminus  \{x\})^*$, by  $  \overline{x} \overline{v} x$ or
vice-versa. We write it for short
\begin{equation}\label{eq:block_transf}xv \overline{x} \rightarrow
\overline{x} \overline{v} x   \quad \mbox { or }  \quad
\overline{x} \overline{v} x\rightarrow  xv \overline{x}  \quad (x
\in {\mathcal A}, \  v \in ( {\mathcal A} \setminus
\{x\})^*).\end{equation} Two finite spinned words $w,w'$ are said
to be \emph{block-equivalent}, and we write  $  w \equiv w'$, if
we can pass from one to the other by a (possibly empty) chain of
block-transformations.  The block-equivalence is an equivalence
relation over spinned words, and  if $ w \equiv  w'$, then $w$ and
$w'$ are spinned versions of a common word over ${\mathcal A}$.

A finite spinned word $w$ is said to be
  in \emph{normal form}
 if $w$  has no factor in $ \cup_{a \in {\mathcal A}} \overline{a} \overline{\mathcal A}^*
 a$.  The normal form is unique, and any sequence of block transformations of the form $\overline{xv}x\to xv\overline{x}$ eventually yields the normal form.
By \emph{block-normalization} of a spinned word (we also say {\em normalization} for short), we   mean a succession
of block-transformations  of the form $\overline{xv}x\to xv\overline{x}$  that produces its normal  form. The
interest of this notion comes from the  following theorem.
\begin{theorem}\cite[Theorem 3.1]{GlenLeveRich:09}
 \label{theo:spinned}
 Let $w$ and $w'$ be two spinned words over ${\mathcal A}  \cup \overline{{\mathcal A}}$. One has $\psi_w =\psi_{w'}$ if and only if $ w  \equiv w'$.
 \end{theorem}

In particular, if
$$ \tilde{\psi}_{a_1} \circ \tilde{\psi}_{a_2} \circ \ldots \circ \tilde{\psi}_{a_k}=  \tilde{\psi}_{b_1} \circ \tilde{\psi}_{b_2} \circ \ldots \circ \widetilde{\psi}_{b_{\ell}},$$
where, for  all $i  ,j$ with $i\in \{1,\cdots,k\}$, $ j \in \{ 1,
\cdots, \ell\}$,  one has   $\widetilde{\psi}_{a_i} \in \{
{\psi}_{a_i}, \overline{\psi}_{a_i}\}$,
 $\tilde{\psi}_{b_j} \in \{ {\psi}_{b_j}, \overline{\psi}_{b_j}\}$,  and $a_i,$ $ b_j \in {\mathcal A}$,
 then $k =\ell$ and $a_i =b_i$, for all $i$.

We recall that pure  episturmian substitutions are the  episturmian substitutions  that have no
 permutation in their decomposition. Any pure episturmian substitution  has a unique directive word in normal form  by  \cite[Lemma 5.3]{GlenLeveRich:09}. Now  consider the  general case of episturmian substitutions (not
only pure ones).
 Let $\theta$ be a  permutation on the alphabet ${\mathcal A}$. It is readily verified that for all $ a \in {\mathcal A}$, one has
\begin{equation} \label{eq:permut}
\theta \circ  \psi_a= \psi _{\theta(a)} \circ \theta, \qquad
\theta \circ  \overline{\psi}_a=\overline{ \psi} _{\theta(a)} \circ \theta.
\end{equation}
According to this property,  if $\sigma$ is an   episturmian substitution, then  $\sigma$
admits a unique decomposition   as
$\sigma= \psi_w \circ \theta_{\sigma}$,
where $\psi_w$ is  a pure    episturmian substitution and $\theta_{\sigma}$ is  a permutation.

 \begin{definition}\label{defi:normal}
  The {\em normal decomposition }  of the episturmian substitution $\sigma$ is defined as the (unique) decomposition
$$\sigma=  \widetilde{ {\psi}}_{a_1}  \circ \cdots \circ   \widetilde{ {\psi}}_{a_n} \circ \theta_{\sigma},$$
with $ \widetilde{ {\psi}}_{a_i} \in \{{ {\psi}}_{a_i} ,\overline{ {\psi}}_{a_i} \}$ for  all $i$, and
where the spinned word $w_{\sigma}$ defined as  $w_{\sigma}=a_1 \cdots a_n$ is in normal
form.
The spinned word $w_{\sigma}$ is called the
{\em normalized directive word} of $\sigma$ and the permutation
$\theta_{\sigma}$ is called the {\em normal permutation} of
$\sigma$. \end{definition}

  We use the notation
$$[\sigma]=\tilde{a}_1  \cdots   \tilde{a}_n  \theta_{\sigma}$$
 as   a word  over the alphabet  ${\mathcal A} \cup \overline{{\mathcal A}} \cup {\mathcal S}_\A$.
 We  call this word the {\em normal form}  of $\sigma$.
 If $\sigma$ is pure, then $\theta_{\sigma}$ is  equal to the identity, denoted as $\Id$. By abuse of notation, when $\theta_{\sigma}= \Id$, we
 write $[\sigma]=\tilde{a}_1  \cdots   \tilde{a}_n $ for short.
  If $\sigma$ is epistandard, then all the  $  \tilde{a_i} $'s  are equal to  $ {a_i}$.
The {\em block-normalization} of an episturmian substitution
$\sigma= \tilde{ {\psi}}_{a_1}  \circ \cdots \circ \,  \tilde{
{\psi}}_{a_n} \circ  \,\theta_{\sigma}$ refers to the block-normalization   of its
directive word $a_1\cdots a_n$.

In order to define formally the surjective map between the words
on the alphabet ${\mathcal A} \cup \overline{{\mathcal A}} \cup
{\mathcal S}_\A$ and corresponding episturmian substitutions, we
introduce the morphism $\mu$ which maps letters to substitutions
with $\mu: \tilde{a} \mapsto \tilde{\psi_{ a}},  \ \theta  \mapsto  \theta $. 
 One has $\sigma=
\mu([\sigma])$ and  $\sigma \circ
\theta_{\sigma}^{-1}=\mu (w_{\sigma})$, where  $w_{\sigma}$ stands for the normalized directive word of $\sigma$. The {\em length} of a spinned word $w$ corresponds
to the usual  notion of length over the alphabet ${\mathcal A}  \cup \overline{{\mathcal A}}$. The {\em length of an episturmian substitution}
 is  defined  as the length of any of its directive words,  denoted as
 $|\sigma|$. It is well defined  by Theorem~\ref{theo:spinned}.
We stress the fact that the length  of $ \sigma=\psi_{a_1} \circ  \cdots \circ  \psi_{a_k} \circ  \theta_{\sigma}$ is   $k$.


The main issue we will have to face comes from the fact that
normalization does not behave well with respect to taking powers: the normalized form of a power of a spinned word in
normal form  is  a priori  not the  power of its normal form.
Indeed, consider the   spinned word $ a \overline{a}$. Its square
$a \overline{a}a \overline{a}$  is  not in normal form because of
the occurrence of
 $\overline{a}a$. This issue  will  be handled in detail in Sections \ref{sec:normalization} and  \ref{sec:lemma}.

The following example shows  that an episturmian substitution may involve only one letter in its  normalized  directive word.

\begin{example}\label{Fibonacci} Consider   the Fibonacci substitution $\sigma$ defined  over the two-letter alphabet $\{a,b\}$ as
 $\sigma  : a \mapsto ab, \ b \mapsto a$. This is  the most classical  example of a  Sturmian substitution, i.e., of a two-letter episturmian substitution.
Over  $\{a,b\}$, one has $\psi_a:  a \mapsto a, \ b \mapsto ab$.
Let  $\theta_{ab}$ be the two-letter permutation that exchanges letters, i.e., $\theta_{ab}:   a \mapsto b, \ b \mapsto a$.
One has $\sigma= \psi_a \circ  \theta_{ab}$,  and  its normalized directive word  $w_{\sigma}$  is $a$. \end{example}




\section{Main result and strategy}\label{sec:result}
Let us  state  and describe  the general  proof for the main
result of this paper.  We recall that Theorem \ref{theo:power}
extends the result from \cite[Theorem 15]{Krieger:08}
  which   proves rigidity for  strict epistandard  words, i.e.,  characteristic  Arnoux-Rauzy
  words.

\begin{theorem}\label{theo:power}
Arnoux-Rauzy words generated by substitutions are rigid.
 \end{theorem}

The proof  of Theorem \ref{theo:power} is based on the following
results. We first  know from  \cite{BerthDDLP:17} the following theorem,
whose proof   relies  on the notion of return words. Note that
this   statement implies that the Perron-Frobenius eigenvalues of
$\sigma$ and $\tau$ are multiplicatively dependent, which is also
a consequence of Cobham's Theorem~\cite{Du:11}.

\begin{theorem}   \cite[Theorem 9]{BerthDDLP:17} \label{theo:commonpower}
Let $x$  be an Arnoux-Rauzy word that is a fixed point of both
$\sigma$ and $\tau$ primitive substitutions. Then there exist $i,j
\geq 1$ such that $\tau^i = \sigma^j$.
 \end{theorem}

We then  prove   primitivity  for  the  substitutions that fix Arnoux-Rauzy words.
\begin{lemma}\label{lem:primitive}
Let $x$ be an Arnoux-Rauzy word   over the  alphabet ${\mathcal A}$ and  let $\sigma$ be a substitution such that  $\sigma(x)=x$,  with $\sigma$ not equal to the identity.
Then $\sigma$  is  primitive. \end{lemma}
\begin{proof}
The  infinite word $x$ is uniformly recurrent due to being an
Arnoux-Rauzy word. The substitution $ \sigma$ is episturmian by Theorem \ref{theo:epi2}.
 We now prove   by contradiction  that
$(|\sigma^n(a)|)_n$ tends to infinity  for each letter $a  \in
{\mathcal A}$. Suppose indeed  that for some letter $a$, the sequence of lengths
$(|\sigma^n(a)|)_n$ does not tend to infinity. With the notation of
Definition \ref{defi:normal}, consider $n$ such that
$\theta_{\sigma}^n=\Id$, and take $\sigma'=\sigma^n$. We have
$\theta_{\sigma'}=\Id$. We also have that $(|(\sigma')^\ell(a)|)_\ell$ does
not tend to infinity. From the normal form of episturmian
substitutions,   one deduces that  no power  of $\sigma'$ contains
some   $\tilde{\psi_b}$ in its normal form, with  $b\neq a$:
otherwise, applying $\sigma'$ to a word adds at least one letter
for each occurrence of $a$ in the word. Since $\sigma \neq \Id$, this is only possible when
$\sigma'= \psi_a^i \circ \overline{\psi_a}^j$, but this
substitution does not have an infinite fixed point, which gives us
the desired contradiction. We thus deduce  from the fact that
$(|\sigma^n(a)|)_n$ tends to infinity  for each letter $a  \in
{\mathcal A}$ together with  $x$ being uniformly recurrent that
$\sigma$ is primitive from \cite[Proposition 5.5]{Queffelec:10}.
\end{proof}

 Next lemma  provides a commutation property  for  substitutions fixing a  common Arnoux-Rauzy word.
 \begin{lemma}\label{lem:commute}
 Let $\sigma$ and $\tau$ be two   substitutions.
If $x$ is an Arnoux-Rauzy word  such that  $\sigma(x)=\tau(x)=x$,
then   $\tau \circ  \sigma= \sigma  \circ  \tau$.
\end{lemma}

\begin{proof}
Let $x$ be  an Arnoux-Rauzy word over the  alphabet
${\mathcal A}$  such that $\sigma(x)=\tau(x)$, with $\sigma$ and $\tau$  being   substitutions.  According   to  \cite[Fact 3.4]{GlenLeveRich:08}, \cite[Remark 3.5]{GlenJustin} and \cite{JustinPirillo:02},
 there exists an   epistandard  Arnoux-Rauzy word  $y$  that has  the
same language as $x$, and  $\tau'$, $\sigma'$
epistandard substitutions that are conjugate respectively to $\tau$
and $\sigma$
  such that $\sigma'(y)=\tau'(y)$. By
\cite[Theorem 15]{Krieger:08}, there exists  a substitution
$\varrho$ such that $\sigma'=\varrho^t$ and $\tau'=\varrho^s$, by
using   rigidity for fixed points of epistandard   substitutions.
Then, one has  $\sigma' \circ \tau'= \tau'\circ \sigma' =
\varrho^{s+t}$. Now since $\tau$ and $\tau'$ are conjugate, we
have that $|\tau(a)|$ and $|\tau'(a)|$ are of the same length and
moreover are abelian equivalent, i.e., contain the same numbers of
occurrences of each letter. Since $\sigma$ and $\sigma'$ are
conjugate, the images of abelian equivalent words under $\sigma$
and $\sigma'$ are of the same length (in fact, they are also
abelian equivalent), hence $ |\sigma \circ \tau(a)|= | \tau\circ
\sigma(a)|$ for each $a \in {\mathcal A}$. We deduce from $\sigma
\circ \tau(x) =\tau \circ \sigma(x)$ that $\sigma \circ \tau =\tau
\circ \sigma$ by considering the respective images of each letter
in $x$ and by recalling that all the letters of ${\mathcal A}$
occur  in $x$.
 \end{proof}

We  now  want to go from  the existence of a  common power
(Theorem \ref{theo:commonpower})  to a ``prefix property''. This
is the object of next lemma which  is  the main    step in  the
proof of Theorem \ref{theo:power}. This is the analogue of
\cite[Proposition 4.1]{RichommeSeebold:12} which is  proved for
Sturmian words.

\begin{lemma} \label{lem:mainlem}
Let $\sigma$ and $\tau $ be two   episturmian  substitutions
 such that $\sigma^n= \tau^m$, for
$n \geq m \geq 1$. Then there exists an   episturmian substitution $\varrho$  such that
$\tau = \sigma \circ  \varrho $.

\end{lemma}

The proof of this lemma is quite involved; we provide it in
Section \ref{sec:lemma}. Using this lemma, we now can prove
Theorem \ref{theo:power}.

\bigskip

\begin{proof}
Let $x$ be an Arnoux-Rauzy word.  Let $\sigma$, $\tau$ be   two
substitutions distinct  from the identity  such that
$\sigma(x)=\tau(x)=x$. By Theorem \ref{theo:epi2}  they are both
episturmian. By Lemma \ref{lem:primitive},  $\sigma$ and  $\tau$
are primitive. By Lemma \ref{lem:commute},  they commute. By
Theorem \ref{theo:commonpower},  they   have a common power,
i.e.,  there exist $n,m \geq 1 $ such that $\sigma^n=\tau^m$.

 We now  prove by induction on $\max(|\sigma|,|\tau|)$ that if  there exist $n,m \geq 1 $ such that
$\sigma^n=\tau^m$, then there exist integers $k$, $\ell$ and an
episturmian substitution $\varrho$ such that $\sigma=\varrho^k$,
$\tau=\varrho^{\ell}$.  Note that this step is   the analogue of
\cite[Corollary 4.3]{RichommeSeebold:12} which holds for Sturmian
words (i.e., for the two-letter case).

We have  $|\sigma| \geq 1$ and $  |\tau| \geq 1$ since they are
primitive.

 If $|\sigma|=|\tau|=1$, then  $\sigma=\tau$. Indeed, since the lengths are equal to 1, due to Lemma \ref{lem:mainlem} and  to the uniqueness of the normal form,
 there exist a letter $a$ and permutations $\theta_{\sigma}$ and $\theta_{\tau}$ such that $\sigma$ and $\tau$ are  both either  of the form
 $\sigma= \psi_a \circ  \theta_{\sigma}$, $\tau= \psi_a \circ  \theta_{\tau}$,  or they both involve  $\overline{\psi_a}$, with both cases  being  symmetric.
 We assume  that  they both involve $ \psi_a$.
  Now since $x=\sigma(x)=\tau(x)$, we have $\psi_a(\theta_{\sigma} (x)) = \psi_a (\theta_{\tau}
 (x))$. Notice that the form of the substitution $\psi_a$ implies that if $\psi_a(w)=\psi_a(w')$ for
some  infinite words $w$ and $w'$, then $w=w'$ (since $\psi_a$ only inserts a letter $a$ before each letter which is not
 $a$). So, $\theta_{\sigma} (x) = \theta_{\tau}
 (x)$.  Since $\theta_{\sigma}$ and $\theta_{\tau}$
 are permutations, it follows that $\theta_{\sigma}=\theta_{\tau}$.

Suppose w.l.o.g. that $n\geq m$ and that the induction property
holds for substitutions with lengths smaller than
$\max(|\sigma|,|\tau|)$ which is equal to $ |\tau|$ (since
$n|\sigma|= m | \tau|$ by the definition of the length of
substitution). By Lemma \ref{lem:mainlem}, there exists an
episturmian substitution $\varrho$ such that   $ \tau=\sigma \circ
\varrho$. From  $ \tau \circ \sigma=\sigma \circ \tau$, we deduce
that $\sigma \circ \varrho \circ \sigma= \sigma \circ \sigma \circ
\varrho$,   and thus $\varrho \circ \sigma= \sigma \circ \varrho$
by left-cancellativity (see Proposition \ref{prop:cancel}). If
$\varrho=\Id$, then $\sigma=\tau$, which provides the desired
induction  property. The case  where $\varrho $ a permutation is similar
to the case where both substitutions have  length 1 above. We now
assume $|\varrho| \geq 1$, and thus $|\tau| > \max
(|\sigma|,|\varrho|)$, which yields in particular $m < n$. One has
$ \tau^m= \sigma^m \circ \varrho^m= \sigma^n$, and consequently
$\sigma ^{n-m}= \varrho ^ m$ (again by left-cancellativity). We
now can  apply the induction hypothesis to $\sigma$ and $\varrho$.
Hence there exist $\varrho'$ episturmian substitution, $k$, $\ell$
integers such that $\sigma=\varrho'^k$, $\varrho=
(\varrho')^{\ell}$, which  also yields    $\tau=
(\varrho')^{k+\ell}$. This   ends the induction proof.
\end{proof}

\section{Block-normalization of powers} \label{sec:normalization}
We recall that the normalized form of a square (power) of a
spinned word in normal form  may not be equal to the square
(power) of its normal form. We study in this section how
normalization behaves with respect to  powers $\sigma^n$ for an
episturmian substitution $\sigma$, by describing a  normalization algorithm in Section \ref{subsec:algo}.
Examples are given in Section \ref{subsec:examples}, and
 more precise
statements are provided in Sections \ref{subsec:secondlevel},
\ref{subsec:moreIII} and \ref{subsec:morenormalform}.  Lastly, the normalization of $[\sigma]^n$ in the case where the normalized
directive word of  $\sigma$ contains
only one letter is  handled in Section \ref{subsec:oneletter}.

\subsection{First notation}

We recall that  $[\sigma]=\tilde{a}_1\cdots
\tilde{a}_k\theta_{\sigma}$  is the word over the alphabet
${\mathcal A} \cup \overline{\mathcal A} \cup
 S_\A$  representing
$\sigma$ in its normal form as defined in Section
\ref{subsec:normalform}. From now on we will keep the notation $k$
for the length of $\sigma$.  We assume $k \geq 1$.  The normalized directive word of
$\sigma$ is    $w_{\sigma}=\tilde{a}_1\cdots \tilde{ a}_k$ and
  $\theta_{\sigma}$ is the normal permutation of $\sigma$.
We  also recall that the letters  $\tilde{a}_i$ belong to  $
{\mathcal A} \cup \overline{ {\mathcal A} }$.



 Let $n$ be a positive integer. Let us now consider $\sigma^n$.    We want to compare the words  $[\sigma]^n$ and $[\sigma^n].$
As mentioned above (see also Example \ref{ex:typeII} and \ref{ex:typeIII}),  the word $[\sigma]^n$ does not have to be the normal form
of $\sigma^n$.  In order to get the normal form $[\sigma^n]$ of
 $\sigma^n$,  we first  handle the occurrences of permutations, which we  shift to the end,
 and secondly, we  perform block-transformations  (see Section \ref{subsec:algo}). The latter  is the  main step that  we  will   have to  work out.

The first piece of the normalization  algorithm   thus   consists in handling the permutation. To get the normalized form of $[\sigma^n]$, we start by shifting
the $n$ occurrences of the permutation $\theta_{\sigma}$ to the
end of the word using the relations (\ref{eq:permut}). In
particular, the  normal permutation of $\sigma^n$ is equal to  the
$n$-th power $\theta_{\sigma}^n$ of the  normal permutation of
$\sigma$.  Recall that the permutation $\theta_{\sigma}^n$ is seen both as a permutation and as a letter.
 This gives a word $Z_{\sigma^n}$ over the
alphabet ${\mathcal A} \cup \overline{ {\mathcal A} } \cup S_\A$,
which is  a priori not the normal form $[\sigma^n]$, and which is defined by
$$Z_{\sigma^n}=
 \tilde{a}_1\cdots \tilde{a}_k
\widetilde{\theta_{\sigma}(a_1)}\cdots
\widetilde{\theta_{\sigma}(a_k)} \cdots
\widetilde{\theta_{\sigma}^{n-1}(a_1)}\cdots
\widetilde{\theta_{\sigma}^{n-1}(a_k)}\theta_{\sigma}^n,
$$
with  $$\sigma^n = \widetilde{\psi}_{a_1}\cdots
\widetilde{\psi}_{a_k}\widetilde{\psi}_{\theta_{\sigma}(a_1)}\cdots
\widetilde{\psi}_{\theta_{\sigma}(a_k)} \cdots
\widetilde{\psi}_{\theta_{\sigma}^{n-1}(a_1)}\cdots
\widetilde{\psi}_{\theta_{\sigma}^{n-1}(a_k)}\theta_{\sigma}^n =
\mu (Z_{\sigma^n}).$$  One has
$$\sigma^n= \mu([\sigma^n])= \mu (Z_{\sigma^n}).$$

If  $n=1$,  then $Z_{\sigma}= [\sigma]$.

Observe that  the values of spins   (i.e.,  the  occurrence  or
not of  a  bar  for $\widetilde{a_i}$)   are  $k$-periodic in
$Z_{\sigma^n}$. A factor of length $k$ that occurs at an index
congruent to $1$ modulo $k$ in $Z_{\sigma^n}$ is called a {\em
$k$-period} of $Z_{\sigma^n}$; it has the form $
\widetilde{\theta_{\sigma}^i (a_1)} \cdots
\widetilde{\theta_{\sigma}^{i}(a_k)}$ ($i=0, \cdots, n-1$)  .
 We use vertical bars to mark periods, which gives
$$Z_{\sigma^n}=
 \tilde{a}_1\cdots \tilde{a}_k \mid
\widetilde{\theta_{\sigma}(a_1)}\cdots
\widetilde{\theta_{\sigma}(a_k)} \mid  \cdots \mid
\widetilde{\theta_{\sigma}^{n-1}(a_1)}\cdots
\widetilde{\theta_{\sigma}^{n-1}(a_k)}  \mid
\theta_{\sigma}^n.
$$
Vertical bars are  thus  located between letters with indices $kj$ and
$kj+1$. Moreover periods are labeled with indices $j$ running from $1$
to $n$.

\begin{remark}\label{remarksbis}
\begin{enumerate}
\item \label{rem:4}  Let $n,m$ be positive integers, and let $\sigma$, $\tau$ be episturmian substitutions such that $\sigma^n=\tau^m$.  At each index
of $Z_{\sigma^n}$ and $Z_{\tau^m}$, we have $\tilde{b}$ for the
same $b\in\mathcal{A}$ but possibly with different spins.

\item  One has $Z_{\sigma^n } \equiv [\sigma^n]$.
\end{enumerate}
\end{remark}

\subsection{The normalization algorithm}\label{subsec:algo}
The normalization of $[\sigma]^n$ in the case where the normalized
directive word of the episturmian substitution $\sigma$ contains
only one letter $\tilde{a}$, i.e., all its letters are equal to
$\tilde{a}$ is provided in Section \ref{subsec:oneletter}.

 We now consider the case  where $[\sigma]$ contains at least two distinct  letters.
  If the word $Z_{\sigma^n}$  is not  the normal form of $\sigma^n$, then  this means  that there are occurrences of non-normal spinned words
of the form $\overline{a }\overline{b_1}
\cdots \overline{b_\ell} a$  in   $Z_{\sigma^n}$, with  the letter  $a$ and  the letters   $b_i$'s  in ${\mathcal A}$. This leads us to introduce  the following definition.
We use the term ``simple''  in next definition   in order to avoid any confusion with  errors  called ``propagated''  below.

\begin{definition}[Simple error]
A \emph{simple $a$-error} is  an occurrence in
$Z_{\sigma^n}$   of a  factor of the form  $$ \overline{a}
\overline{b_1} \cdots \overline{b_\ell} a,$$
with  the letter  $a$ and  the letters   $b_i$'s  in ${\mathcal A}$ (and  the block $b_1\cdots b_\ell $ being possibly empty).  The letter $a$ is
referred  to as the {\em letter of the error}. A simple error is an $a$-error for some $a$ in ${\mathcal A}$.
\end{definition}

We assume in all this section  that there exists a simple
error in $Z_{\sigma^n}$.
 The leftmost occurrence of a simple error plays a
significant role, as stressed by Lemma \ref{lem:left}  (see Section \ref{subsec:secondlevel} below), which
basically says that all simple errors in $Z_{\sigma^n}$ are
defined by the leftmost error,  and that  there is only one possible simple error at the boundary of two periods. 

After correction of the simple errors using block-transformations
\eqref{eq:block_transf}, new simple errors might occur (see
Example \ref{ex:typeII}). This yields  the following  definition.
\begin{definition} [Propagated error of  the  first level] \label{def:maximal}
 Let $a \in {\mathcal A}$ and let  ${{\mathcal A} '}=
{ {\mathcal A} \setminus \{a\}}$.
A  {\em
propagated error of  the  first level}  is  defined as
an occurrence of  a factor  of $Z_{\sigma^n}$  of  the form
 \begin{equation}\label{eq:max_factor} \overline{a} ( \overline{{\mathcal A}')
}^* \overline{a}( \overline{{\mathcal A}') }^*  \overline{a}
\cdots \overline{a} ( \overline{{\mathcal A}') }^*  a(
\overline{{\mathcal A}') }^* a\cdots  a( \overline{{\mathcal A}')
}^* a = ( \overline{a} ( \overline{{\mathcal A}')}^*)^+ ( a
\overline{{\mathcal A}')}^*)^*a \end{equation}  that can neither
be extended  in $Z_{\sigma^n}$  to the left   by
$\overline{a}(\overline{{\mathcal A}'})^*$, nor to the right by
$(\overline{{\mathcal A}'})^*a$. We often write {\em propagated
$a$-error of the first level} referring to the role played by the
letter $a$.
\end{definition}

Indeed, after applying a block-transformation $\overline{a} (
\overline{{\mathcal A}') }^*  a \to a ( {\mathcal A}')^*
\overline{a}$ to the only simple error here, we encounter two new
simple errors immediately to the left and to the right from the
initial one, and so on inside the factors of the form
\eqref{eq:max_factor}. The terminology ``first level'' in the
above definition refers to the fact that there will be  possibly
another level of block-transformations if after correcting the
errors inside factors of the form \eqref{eq:max_factor} yet new
errors occur (see Example \ref{ex:typeIIIbasic}).

 We now   describe   the algorithm allowing one  to get the normal
form of $[\sigma]^n$  (which is  also the normal form of
$\sigma^n$),  when the normal form of $\sigma$ admits at least two distinct
 letters, $\sigma$ is assumed to have  at least one simple error
and  $ n\geq 2$. We  first move the permutation to the end. This
gives the word $Z_{\sigma^n}$.  We then correct the simple errors
simultaneously using block-transformations. The fact that we can
correct them independently comes from Lemma \ref{lem:left}.
If the outcome is  not in normal form, this means
  that simple errors     propagate
when  performing block-transformations, that is, there exist  propagated errors of the first level.
 According to Lemma  \ref{lem:leftmost},  the  supports of propagated errors of the first level  do not intersect from  a period to the next one   and  the block-normalization is
 performed  at the boundary of each  pair of successive $k$-periods.
 Now, let    $Z'_{\sigma^n}$ stand for  the word
obtained after the first level of block-transformations.
The exact form of $Z'_{\sigma^n}$
 is provided in Lemma \ref{lem:propa}.
 If
$Z'_{\sigma^n}$ is not in normal form  (see Example
\ref{ex:typeIII}), we normalize it with the {\em second level of
block-transformations} consisting in this last round of
block-transformations to be performed.    Assertion \ref{III:vi}
from Lemma \ref{lem:type3} provides the details of the  second
level of block-transformations to be performed  which   are proved
to   lead to the  normal form of $[\sigma^n]$.

By Lemma \ref{lem:type3pre},  after the first level of
block-transformations, if a new simple error is created in the
$j$-th $k$-period,  with $j=2, \dots, n-1$, then
 new simple errors are created in each $k$-period of index $j$,  with  $j=2, \dots, n-1$ (all periods except for the first one and the last one). This justifies the following definition.

\begin{definition}[Errors of type I and II]\label{def:typeII}
Let  $Z'_{\sigma^n}$ be the word   obtained  after  performing  the first level of block-transformations.
Propagated  errors of the first level are  said   of type I if  $Z'_{\sigma^n}$ is in  normal form, i.e.,  $Z'_{\sigma^n}=[\sigma^n]$.
They are  said of type II otherwise.
In other words,  all the propagated  errors of the first level   of the  word $Z_{\sigma^n}$   are  of type II if $Z'_{\sigma^n}$ is not  in  normal form.
For sake of simplicity in the terminology,     an error  of type I (resp. II) refers to a   propagated  error of the first level of type I (resp. II).
We also say that  $Z_{\sigma^n}$ has    errors  of type I  (resp. II).\end{definition}

We will rely on the  following notation.
\begin{notation}\label{notation}
The  notation $k$ stands  for
 the length of $\sigma$,   $i$ for the  index of occurrence of the  leftmost
propagated error of the first level in $Z_{\sigma^n}$,   and $k+p$  for the  last index of the leftmost
propagated error of the first level, i.e.,  $[i,k+p]$ is the support of  the leftmost
propagated error of the first level.  The letter  of the leftmost
error (if any)  is  denoted by $a$.
\end{notation}

 The second level of   block-transformations
is described in detail in  Lemma
\ref{lem:type3}, Assertion \ref{III:vi} below. Let us give a brief overview of this second level.
One has  $Z'_{\sigma^n}[i+kj]=a$ and
$Z'_{\sigma^n}[p+k+kj]=\overline{a}$ for $j=0, \dots k-2$, by
Lemma \ref{lem:propa}, Assertions \ref{IIii} and \ref{IIiii}.
By Lemma \ref{lem:type3pre},
simple errors that occur in $Z'_{\sigma^n}$
when errors are of type II are factors of the form
$\overline{a} \overline{b_1}\cdots\overline{b_t} a$ occurring at
positions $[p+kj, i+kj]$ for $j=1,\dots, k-2$. These simple errors
can also propagate to blocks of $\tilde{a}$ of length greater than
2 and including positions $i+kj$ and $p+kj$. The second level
of block-transformations consists in normalizing these errors,
after which we obtain the normal form of $\sigma^n$.

\subsection{Examples}\label{subsec:examples}

After correction of the simple errors using block-transformations
\eqref{eq:block_transf}, new simple errors might occur. We first
give an example that illustrates the first level
block-transformations that  has to be performed.

\begin{example}\label{ex:typeII}
Let  $\sigma$ with normal form $$[\sigma]=a \overline{b} a c
\overline{ad ae }$$
with $a,b,c,d,e$ being
distinct letters and   $\theta_{\sigma}$   being  the identity.


In order to  normalize $Z_{\sigma^3}$ with
$$Z_{\sigma^3}=a \overline{b} a c
\overline{a}\overline{d} \tcb{\overline{a}\overline{e} \mid  a}
\overline{b} a c \overline{a}\overline{d}
\tcb{\overline{a}\overline{e} \mid a} \overline{b} a c
\overline{a}\overline{d} \overline{a}\overline{e} ,$$ first, we  fix   the  simple $a$-errors   that occur  at the boundaries of  the  periods,
i.e., occurrences of $\overline{ae}| a.$ Their letters   are  marked with blue.
 Their  block-normalization  gives
$$Z_{\sigma^3}=a \overline{b} a c
\overline{a}\overline{d} \tcb{\overline{a}\overline{e} \mid  a}
\overline{b} a c \overline{a}\overline{d}
\tcb{\overline{a}\overline{e} \mid a}
\overline{b} a c
\overline{a}\overline{d} \overline{a}\overline{e}  \equiv  a \overline{b} a c
\overline{a}\overline{d}
\tcb{{a}{e} \mid \overline{a}}
\overline{b} a c \overline{a}\overline{d}
\tcb{{a}{e} \mid \overline{a}}  \overline{b} a c
\overline{a}\overline{d} \overline{a}\overline{e} .$$

 We see that, after having  normalized the  simple errors, new  simple errors occur  (i.e., occurrences of  factors $\overline{a} \overline{d} a $    and  $\overline{a} \overline{b} a $), which induces a
 propagation  in the  block-transformations to be performed; the involved letters    are marked with red:
    $$Z_{\sigma^3} \equiv a \overline{b} a  { c}
{\color{red}\overline{a}\overline{d}
{a}{e} | \overline{a} \overline{b}
a} c{\color{red}\overline{a}\overline{d}
{a}{e}| \overline{a} \overline{b}
a } c \overline{a}\overline{d}
\overline{a}\overline{e}.$$

 After performing the corresponding  block-normalization,  and by  noticing that
the letter $c$  with spin $0$ stops the propagation,  we get
the normal form of $\sigma^3$:

$$Z'_{\sigma^3}=[\sigma^3]= a \overline{b} a  {c}
{\color{red}a d \overline{a}e | a b
\overline{a}} {c} {\color{red}a d
\overline{a}e | a b \overline{a} } c
\overline{a}\overline{d}
\overline{a}\overline{e}.$$

\end{example}

\begin{example}[Error of type I with a permutation]\label{ex:typeIIperm}
Let  $\sigma$ with normal form $$[\sigma]=c \overline{b} a
c\overline{b} c \overline{a} \overline{b} \  \theta_{\sigma} $$
with $a,b,c$ being distinct letters,  and  $\theta_{\sigma}$ being
the permutation $\theta_{ac}$ exchanging $a$ and $c$.

We have

$$Z_{\sigma^3}= c \overline{b} a
c\overline{b} c {\color{red} \overline{a} \overline{b} \mid a }
\overline{b} c a\overline{b} a {\color{red} \overline{c}
\overline{b} \mid c } \overline{b} a c\overline{b} c \overline{a}
\overline{b}  \mid \theta_{\sigma} ^3.
$$

This word contains two  propagated errors of the first level; the letters of their  supports  are marked with red: the
first one is of letter $a$, and the second one is of letter $c$.
Normalizing these errors, we get the normal form $[\sigma^3]$ of  $Z_{\sigma^3}$:

$$Z'_{\sigma^3}=[\sigma^3]= c \overline{b} a
c\overline{b} c {\color{red} a b \mid \overline{a} } \overline{b}
c a\overline{b} a {\color{red} cb \mid \overline{c} }\overline{b}
a c\overline{b} c \overline{a} \overline{b} \mid \theta_{\sigma}^3.
$$

In this example  $k=|\sigma|=8$,  and the support  of  the
leftmost  $a$-error of type I  is $[7,9]$, with $i=7$, $p=1$.
\end{example}

The following  example is  an illustration of the second level of  block-transformations.
 \begin{example}   \label{ex:typeIIIbasic} Let   $\sigma$ with  $[\sigma]=a
\overline{bca} $, with $a,b,c$   distinct letters  and the  normal permutation $\theta_{\sigma}$ being  equal to the identity.
Letters in red  below indicate simple errors that occur at the boundary between two periods (and they coincide with the  leftmost   propagated errors of  the  first level), whereas   letters in blue  refer to
simple errors created after performing this first level of  block-transformations  on  letters in red.  One has
$$\sigma^3=  Z_{\sigma^3}= a
\overline{bc} {\color{red} \overline{a}} \mid  {\color{red}{ a}}
\tcb{\overline{bc}}  {\color{red} \overline{a} } \mid
{\color{red}{ a}} {\color{blue}\overline{bc}} \overline{a}  ,$$
$$Z'_{\sigma^3}=  a  \overline{bc}   {\color{red}{a}} \mid  {\color{red}\overline{a}}     {\color{blue}\overline{bc} }  {\color{red}{a}} \mid  {\color{red}\overline{a}}  {\color{blue}  \overline{bc}}  \overline{a}.$$
We notice that  $Z'_{\sigma}$  is not in normal form because of the factor ${\color{red}\overline{a}}     {\color{blue}\overline{bc} }  {\color{red}{a}}$.
Lastly  we get for the normal form  $[\sigma^3]$ of $\sigma^3$
$$[\sigma^3]=a \overline{bc} a \mid
a  {bc} \overline{a} \mid
\overline{a}
 \overline{b}\overline{c}
\overline{a} .$$
\end{example}

\begin{example}[Error of type II]\label{ex:typeIII}

Let $\tau$ with normal form $$[\tau] =a \overline{b} a
\overline{c} \overline{a}\overline{d} \overline{a}\overline{e} $$
with pairwise distinct letters $a,b,c,d,e$ and $ \theta_{\tau} $
being  the identity. The difference with $\sigma$ with normal form
$[\sigma]=a \overline{b} a c \overline{ad ae }  $   (handled
previously  in Example \ref{ex:typeII}) is that the letter $c$
occurs  with a  spin equal to $1$  (i.e., as $\overline{c})$. Let
us normalize $\tau^3$. The first  level is   the same as for
$\sigma$ (up to bars on $c$'s). We thus first normalize the blocks
$\overline{a}\overline{e} | a$. The letters  in the supports of
the  propagated errors of the first level  are  marked by red;
the letter that causes the  second level  of  block-transformations
is marked by blue (before and after normalization):

$$Z_{\tau^3} = a \overline{b} a \overline{c}
{\color{red}\overline{a}\overline{d}
\overline{a}\overline{e} | a \overline{b}
a } {\color{blue}\overline{c}} {\color{red}
\overline{a}\overline{d}
\overline{a}\overline{e}| a \overline{b}
a } \overline{c} \overline{a}\overline{d}
\overline{a}\overline{e} .$$
The first level  of  block-transformations gives as before:
$$Z'_{\tau^3}= a \overline{b} a \overline{c}
{\color{red}a d \overline{a}e | a b
\overline{a} } {\color{blue}\overline{c}}
{\color{red}a d \overline{a}e | a b
\overline{a} } \overline{c}
\overline{a}\overline{d}
\overline{a}\overline{e} .$$

This  creates  the   $a$-error $  \overline{ac}  {a}$, and the
second  level  of  block-transformations produces the normal form $[\tau^3]$ of $\tau^3$:

$$[\tau^3] = a \overline{b} a \overline{ c}
{\color{red}a d \overline{a}e | a b
a } {\color{blue}c} {\color{red}\overline{a} d
\overline{a}e | a b \overline{a} }
\overline{c} \overline{a}\overline{d}
\overline{a}\overline{e} .$$

One has $ k=8$, $i=5$, $p=3$. We see that  the    full propagation of the errors during the normalization process involves letters whose supports form an interval that overlaps all the periods, namely $[i,p+(n-1)k]=[5,19]$.

\end{example}

%
%
%
%

\subsection{First level of  block-transformations}\label{subsec:secondlevel}
Let $\sigma$ be an episturmian substitution of positive  length $k$.
 We  first prove prove that all simple errors in $Z_{\sigma^n}$ are
defined by the leftmost error,  and that there is at most one
simple error at each boundary of two periods.

\begin{lemma} \label{lem:left}
Let $\sigma$ be an  episturmian substitution such that   $Z_{\sigma^n}$ contains   a simple error, with $ n\geq 2$. The occurence of the  leftmost  simple error  is at
the boundary between the first and second periods.   In particular,
it contains in its support the integer $k= |\sigma|$. Let $a$ stand for its letter.
For any $j$ with $1\leq j \leq n-1$,  one also  gets at the boundaries   of  the $j$-th   and the $j+1$-th
periods  a  simple $\theta_{\sigma}^{j-1} (a)$-error
 in $Z_{\sigma^n}$. Moreover, the supports of the simple errors  in
$Z_{\sigma^n}$ are $k$-periodic, 
 and there is only one
possible simple error at the boundary of two  periods.
\end{lemma}

\begin{proof}
Simple errors  cannot occur inside
the $k$-periods (blocks of length $k=|\sigma|$), since the word
$Z_{\sigma^n}$ is a concatenation of words of length $k$ in normal
form. Due to $k$-periodicity of bars and letters (modulo the
permutation $\theta_{\sigma}$), simple errors occur with
period $k$. So, the support of the  leftmost  simple error  crosses the boundary
 between the first and second periods. The fact that there is only one
possible simple error at the boundary of two  periods comes from the shape of simple errors.
\end{proof}

We now prove the  important property  that,  in the first level,
 propagated errors do not affect each other from one $k$-period to another, since their supports do not intersect (see Lemma \ref{lem:leftmost}).
Consequently, if there are no more simple errors after this first level  of  block-transformations, then
 the changes in the
process of normalization are local within the  intervals $[i+kj,
p+k(j+1)]$ ($j=0,\ldots,{n-2}$). In particular, for the
corresponding substitutions, one gets $$\mu(\pref_{p+k}
([\sigma^n]))= \mu(\pref_{p+k}([\sigma]^n)).$$

 \begin{lemma}\label{lem:leftmost}
Let $\sigma$ be an episturmian substitution of positive length $k$
such that   $Z_{\sigma^n}$ contains   a simple error, with $ n\geq 2$. We consider
 a  propagated error of the first level which
contains $k$ in its support. Let   $i$ stand for the index of its
first  letter  in the word $Z_{\sigma^n}$ and   let $p$ be such
that
 $p+k$ is the  index of its last letter. One has   $ p < i$,  i.e., the  supports of propagated errors of the first level  do not intersect,
 and    the block-normalization is thus
 performed  at the boundary of each  pair of successive $k$-periods.

 \end{lemma}
 \begin{proof} Let $a$ be the letter of  the leftmost error in $Z_{\sigma^n}$.
In $Z_{\sigma^n}$ letters $\tilde{a}$   occur with   spin $0$
from index $k+1$ to $k+p$ (they are not barred), whereas
letters $\tilde{a}$ occur with  spin   $1$   from index $i+k$ to
$2k$, by $k$-periodicity. Hence if the letter of the next error  is also $a$, the next propagated error
of the first level starts  (at $i+k$) after the previous one ends
(at $p+k$).  Consider  now the case where   the letter of the next error  is $b$ with  $b\neq a$.
The  letter $a$  has  spin 0 at position $k+p$, and a  $b$-error  contains only  barred letters  for letters distinct from  $b$. So, the next  propagated error of  the  first level
 also starts at position $k+i>k+p$. This implies that the  supports of propagated errors of the first level  do not intersect from  a period to the next one.
\end{proof}

The  support of  the leftmost
propagated   error of the first level  in $Z_{\sigma^n}$ is thus   the
interval of integers $[i,p+k]$,   with $ 1 \leq i \leq k,$ $1 \leq p  \leq  k$.
There  is  also a  propagated  $\theta_{\sigma}(a)$-error  of  the first level  that starts at index  $i+k$.
So,  we have $i>1$ and $p<k$.
 Similarly,   there  is  also a  propagated  $\theta^{j}_{\sigma}(a)$-error  of  the first level  that starts at index  $i+jk$, for any  $0 \leq j < n-1$.
  The
initial positions of  the propagated errors  of the first level are $i+kj$,
$0 \leq j < n-1$, and their supports are $[i+kj, p+k(j+1)]$, for  $0
\leq j < n-1$. The word  $Z'_{\sigma^n}$, obtained after the first level of block-transformations,
is obtained  after
block-transformations  that are performed inside the supports $[i+kj,
p+k(j+1)]$ ($0 \leq j < n-1$) of the errors of the first level.
The exact form of $Z'_{\sigma^n}$
 is provided in Lemma \ref{lem:propa}   which
 describes  the   changes of spins  performed  during  the first level of block-transformations.

But before  stating  Lemma \ref{lem:propa}, we  first discuss and illustrate  a special case of the notion of a  propagated error of the first
level,  namely when it contains a simple error $\overline{a}a$ (clearly, this is
only possible at the index $k$, since $Z_{\sigma^n}[1,k]$ and
$Z_{\sigma^n}[k+1,2k]$ do not have  simple errors inside them). The
\emph{central part} of a propagated $a$-error of the first level
is then defined as the longest factor of the form $\overline{a}^+
\mid a^+$ that contains $k$ in its support (in particular
$Z_{\sigma^n}[k]=\overline{a}$ and $Z_{\sigma^n}[k+1]=a$).

\begin{example} \label{ex:central}
({\bf{Error of type I with a central part}})


This example  illustrates  Assertion \ref{IIv} in Lemma
\ref{lem:propa} stated below.

Let $\sigma$ with normal form $[\sigma]=a^3  b \overline{a}^3
\overline{c} \overline{a}^3\overline{d} \overline{a}^3 $. By
considering $Z_{\sigma^3}$, we get  errors $\overline{a} a$ of the
first level. The letters of the supports of  the  propagated errors  of the first
level are marked in red below:

$$Z_{\sigma^3} = a^3  b {\color{red}\overline{a}^3
\overline{c} \overline{a}^3\overline{d} \overline{a}^3|a^3} b
{\color{red}\overline{a}^3  \overline{c}
\overline{a}^3\overline{d} \overline{a}^3| a^3} b \overline{a}^3
\overline{c} \overline{a}^3\overline{d} \overline{a}^3 .$$

During  the normalization, first we normalize  the  central part
$\overline{a}^3| a^3$, which gives:

$$\overline{a}^3| a^3 \to \overline{a}^2 a|
\overline{a} a^2\to \overline{a} a \overline{a}| a \overline{a} a
\to  a \overline{a} a | \overline{a} a \overline{a} \to a^2
\overline{a} |  a \overline{a}^2 \to a^3 | \overline{a}^3 .$$

Then we handle its propagation to the left (with the error
involving $d$):

$$ a^3  b {\color{red}\overline{a}^3
\overline{c} \overline{a}^3\overline{d} a^3 | \overline{a}^3} b
{\color{red}\overline{a}^3 \overline{c} \overline{a}^3\overline{d}
a^3 | \overline{a}^3} b \overline{a}^3  \overline{c}
\overline{a}^3\overline{d} \overline{a}^3 \to $$

$$ a^3  b {\color{red}\overline{a}^3
\overline{c} \overline{a}^2 a d \overline{a} a^2 | \overline{a}^3}
b {\color{red}\overline{a}^3 \overline{c} \overline{a}^2 a d
\overline{a} a^2 | \overline{a}^3} b \overline{a}^3  \overline{c}
\overline{a}^3\overline{d} \overline{a}^3 .$$

Now we fix twice the error $\overline{a}a$ in
$\overline{a}^2 a$ (which occurs twice after the  first  two occurrences of the letters $\overline{c}$):

$$ a^3  b {\color{red}\overline{a}^3
\overline{c} a \overline{a}^2  d \overline{a} a^2 |
\overline{a}^3} b {\color{red}\overline{a}^3 \overline{c} a
\overline{a}^2  d \overline{a} a^2 | \overline{a}^3} b
\overline{a}^3  \overline{c} \overline{a}^3\overline{d}
\overline{a}^3 .$$

Now we fix  the error involving $c$:

$$ a^3  b {\color{red}\overline{a}^2 a
c \overline{a}^3  d \overline{a} a^2 | \overline{a}^3} b
{\color{red}\overline{a}^2 a c \overline{a}^3  d \overline{a} a^2
| \overline{a}^3} b \overline{a}^3  \overline{c}
\overline{a}^3\overline{d} \overline{a}^3 $$

 and finally the part $\overline{a}^2 a$ between
 $b$ and $c$, as well as the part $\overline{a}
 a^2$ after  $d$, which gives:

$$[\sigma^3]=a^3  b {\color{red}a \overline{a}^2
c \overline{a}^3  d  a^2 \overline{a}| \overline{a}^3} b
{\color{red}a \overline{a}^2 c \overline{a}^3  d  a^2
\overline{a}| \overline{a}^3} b \overline{a}^3  \overline{c}
\overline{a}^3\overline{d} \overline{a}^3  .$$
\end{example}

Let us now state Lemma \ref{lem:propa}  which characterizes the word $Z'_{\sigma^n}$ obtained after the  first
level of block-transformations. We  assume that
$Z_{\sigma^n}$ contains  a  propagated error of the first level
and that the letter of  its leftmost one is $a$ (with $a \in {\mathcal A}$), and we use  Notation \ref{notation}.
 See also Figure \ref{fig:level2}  below for
an illustration of the case where there is no central part.

\begin{lemma}\label{lem:propa} Let $\sigma$ be an episturmian substitution of length $k \geq 1$.


\begin{enumerate}
\item\label{IIi}    If a letter  $\tilde{b}$, with $b \in {\mathcal A}$ and  $b \neq a$,   occurs  in the support  $[i,k+p]$ of the  leftmost propagated  error  of the first level in   $Z_{\sigma^n}$, then it occurs  with spin $1$ (i.e., as  $\overline{b}$)    in   $Z_{\sigma^n}[i,k+p]$, and  with spin $0$ in $Z'_{\sigma^n}[i,k+p]$.
\item  \label{IIii}  The first   letter   of the  leftmost  propagated error  of the first level   occurs   (at index $i$) as  $\overline{a}$ in $Z_{\sigma^n}$, and  as  $a$ in $Z'_{\sigma^n}$.
\item \label{IIiii} The  last letter    of the  leftmost error  propagated error  of the first level  occurs  (at index $k+p$) as
$a$  in $Z_{\sigma^n}$, and as   $\overline{a}$ in  $Z'_{\sigma^n}$.
 \item \label{IIbis}
 The spins  of the letter $a$  in the leftmost  propagated  error  of the first level are    equal to $1$  before index $k$ included and  to $0$  after index $k$ in $Z_{\sigma^n}$.

\item \label{IIiv}  If the leftmost propagated  error  of the first
level has no central part,
 then
the spins of all the occurrences of $a$ in $[i,k+p]$, except the
first and last ones, coincide in $Z_{\sigma^n}$  and
$Z'_{\sigma^n}$.

\item  \label{IIv}
Otherwise,   the leftmost propagated  error  of the first level
  contains a central part       of the form $\overline{a} ^+ \mid  a ^+$ in $Z_{\sigma^n}$ with  $Z_{\sigma^n}[k]=\overline{a}$ and $Z_{\sigma^n}[k+1]=a$.

Let $r,q,r',q'$ be such that the   central part is of the form
$\overline{a }^r  \mid a ^q$ and  its normalization  is of the
form $a ^{q'}   \overline{a} ^{r'}$.
\begin{enumerate}
\item If  the leftmost  propagated error of the first level  contains only  letters $a$,  then  $r'=r$ and  $q'=q$.
\item  If the   leftmost   propagated error  of the first level contains an occurrence of   a  letter $\overline{b}$  with $b\neq a$   at an index  smaller  than $k$ and
 if it contains only occurrences of the letter $a$
after index $k+1$,  then
 $r'=r+1$  and  $q'=q-1$.
\item  If the  leftmost propagated error  of the first level contains an occurrence
of a letter $\overline{c}$   with $c\neq a$ at an index bigger
than $k$
 and if  it contains only occurrences of the letter $a$
before  index $k$,  then   $r'=r-1$  and  $q'=q+1$.


\item  If  the leftmost propagated error  of the first level  contains an occurrence of a  letter $\overline{b}$  with $b\neq a$   at an index  smaller  than $k$
and an occurrence of a letter $\overline{c}$   with $c\neq a$ at
an index bigger than $k$,   then $r'=r$ and $q'=q$.  Moreover,
except for the central part,  for the first and the last occurrences of
$a$,  the  spins  of the occurrences of the letter $a$ inside the leftmost
propagated error  of the first level remain unchanged.


\end{enumerate}
\item \label{II9} The substitution   $ \mu(\pref_{k+p} Z_{\sigma^n})$
coincides with the substitution $\mu( \pref_{k+p} Z'_{\sigma^n} )$.
In particular, there exists an episturmian substitution  $\varrho$ such that
the substitution  $\mu(\pref_{k+p} Z'_{\sigma^n})$  is equal to
$\sigma \circ \varrho$.
\item \label{II10}  Assume that the errors  are of type I, i.e., by definition, $[\sigma^n]=  Z'_{\sigma^n}$.
If  an occurrence of the   letter $\tilde{a}$ changes its spin in  $[i+1,k]$, then there
exists a   central part,  and   its  index of occurrence   belongs
to the central part. In addition, the changes in the process of
normalization are local within the  interval $[i+kj, p+k(j+1)]$
($j=0, \ldots, {n-2}$). In particular, $\mu(\pref_{p+k}
([\sigma^n]))=\mu(\pref_{p+k}([\sigma]^n))$.
 \end{enumerate}
\end{lemma}
\begin{proof}


The first part of Assertion (\ref{IIi}) and Assertion
(\ref{IIbis}) follow from the definition of the propagated error
of the first level.

The other assertions are derived from the block-normalization process  which works as
follows. First, we consider the  case where there is no  central part. The
block-normalization  of the propagated error of the first level (where the letters $c_{ef}$ and $d_{gh}$ belong to ${\mathcal A} \setminus \{a\}$  according to Definition \ref{def:maximal})

$$ \overline{a} \overline{c_{11}} \cdots
\overline{c_{1m_1}}  \cdots \overline{a}
\overline{c_{t1}} \cdots \overline{c_{tm_t}}
 \overline{a}
\overline{b_1} \mid  \cdots \overline{b_l} a \overline{d_{11}}
\cdots \overline{d_{1s_1}} a \overline{d_{u1}} \cdots
\overline{d_{us_u}}  a  $$ yields
$$
 a {c_{11}} \cdots
{c_{1m_1}}  \cdots \overline{a} {c_{t1}} \cdots {c_{tm_t}}
\overline{a} {b_1} \mid  \cdots {b_l} a {d_{11}} \cdots {d_{1s_1}}
a {d_{u1}} \cdots {d_{us_u}}  \overline{a}.
$$

Indeed, we first fix the simple error  $\overline{a}
 \overline{b_1} \mid  \cdots \overline{b_l} a$ between two
consecutive periods by applying block-normalization. Then, we
consecutively fix in the same way the newly occurred errors $
\overline{a} \overline{c_{j1}} \cdots \overline{c_{jm_j}} a$ to
the left, one  by one, for $j=t, \dots, 1$, and also the  errors  $
\overline{a}\overline{d_{j1}} \cdots \overline{d_{ js_{j}}} a$ to the right,  for $j=1, \dots,u$, as in Example \ref{ex:typeII}. 
The   first and last occurrences of $\tilde{a}$ change their spin
once, hence Assertions (\ref{IIii}) and (\ref{IIiii}). Except for
the first and last  occurrences of $\tilde{a}$, letters
$\tilde{a}$ inside the  propagated error of the first level change their spin twice,
and hence remain unchanged.  So, we proved the second part of
Assertion (\ref{IIi})  and Assertion (\ref{IIiv}).



Now, consider the case where there is a   central part (Assertion
(\ref{IIv})), i.e., we need to normalize

$$ \overline{a} \overline{c_{11}} \cdots
\overline{c_{1m_1}}  \cdots \overline{a} \overline{c_{t1}} \cdots
\overline{c_{tm_t}}
 \overline{a}^r \mid a^q \overline{d_{11}}
\cdots \overline{d_{1s_1}} a \overline{d_{u1}} \cdots
\overline{d_{us_u}}  a  ,$$ similarly  as in Example \ref{ex:central}.

Case (a) corresponds to $t=0$ and $u=0$, and
using several times $\overline{a}a \to a \overline{a}$,  the block-transformations move the
bars to the right and yields
$a^q\overline{a}^r .$

Case (b) corresponds to $t>0$ and $u=0$, and the  block-transformations
(first inside the central part, then in the blocks to left of the
central part, and then again in the central part moving the bars
to the right) yield

$$
 a {c_{11}} \cdots
{c_{1m_1}}  \cdots \overline{a} {c_{t1}} \cdots {c_{tm_t}}
a^{q-1}\overline{a}^{r+1}.
$$

Cases (c) and (d) are proved exactly in the same way. Along the
lines, we also proved  Assertions (\ref{IIi}), (\ref{IIii}) and
(\ref{IIiii}) in the case when we have a central part.

For the proof of  Assertion (\ref{II9}), we use the fact that we can move  the permutations via (\ref{eq:permut}).

Assertion (\ref{II10}) is straightforward.

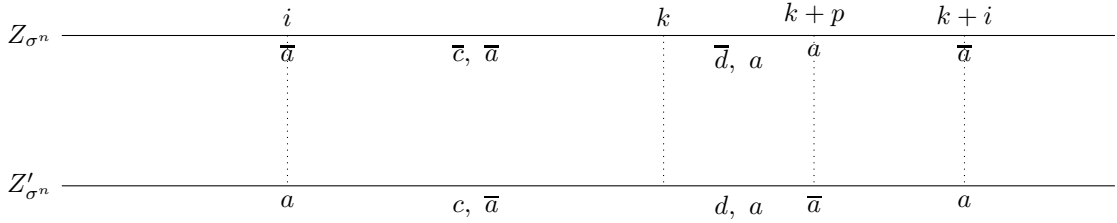
\begin{figure}
\begin{tikzpicture}
\draw  (1,0) -- (15,0);
\draw (4,0) node [above] {$i$};
\draw (4,0) node [below] {$\overline{a}$};
\draw (6.5,0) node [below] {$\overline{c}, \ \overline{a}$};
\draw (6.5,-2) node [below] {${c},\  \overline{a}$};
\draw (4,-2) node [below] {${a}$};
\draw (10,0) node [below] {$\overline{d},\ {a}$};
\draw (10,-2) node [below] {${d},\ {a}$};
\draw (9,0) node [above] {$k$};
\draw (11,0) node [above] {$k+p$};
\draw (11,0) node [below] {${a}$};
\draw (11,-2) node [below] {$\overline{a}$};
\draw (13,0) node [above] {$k+i$};
\draw (13,0) node [below] {$\overline{a}$};
\draw (13,-2) node [below] {${a}$};
\draw [dotted] (4,0)-- (4,-2);
\draw [dotted] (11,0)-- (11,-2);
\draw [dotted] (9,0)-- (9,-2);
\draw [dotted] (13,0)-- (13,-2);
\draw (1,0) node[left]{$Z_{\sigma^n}$};
\draw (1,-2) node[left]{$Z'_{\sigma^n}$};
\draw (1,-2) -- (15,-2);
\end{tikzpicture}
\caption{Illustration   of the transformations  performed  during  the first  level of block-transformations, in  the case where  there  is no central part.} \label{fig:level2}
\end{figure}
\end{proof}

The following technical proposition will be used in the proof of
Lemma \ref{lem:casepropa}. We state it for convenience, although it is a direct
consequence of Lemma \ref{lem:propa}, Assertion \ref{IIv}.
\begin{proposition} \label{prop:central}
 Let $\sigma$ be an episturmian substitution.
We assume that    $Z_{\sigma^n}$ contains  an $a$-error   having  a central part. Let $[i,k+p]$ stand for the  support of its leftmost   propagated error  of the first level.
Let $t$ be the  starting position of the central part  ($i\leq t \leq k$) and  let  $q$ be such  that its ending position is of the form  $k +q$
(with $k+q\leq k+p$), i.e.,
$$Z_{\sigma^n} [t,k]=\overline{a}^{k-t+1}, \quad Z_{\sigma^n} [k+1,k+q]={a}^{q}, \quad
Z_{\sigma^n} [t-1]\neq\overline{a}, \quad Z_{\sigma^n}
[k+q+1]\neq {a}.$$

Let $\ell$ be a
position in  the central part, i.e., $t\leq \ell  \leq k+q$. If
$[\sigma^n][\ell]=a$, then the following holds:

\begin{enumerate}
\item \label{prop:i}  if $t>i$ (the  leftmost   propagated error  of the first level does not start with the central part), then
$q\geq \ell-t+1$;


\item \label{prop:ii}  if $t=i$ (the leftmost   propagated error  of the first level starts with the central part), then
 $q\geq \ell-t$.

\end{enumerate}
\end{proposition}

\begin{proof} The first part follows from Lemma \ref{lem:propa}, Assertion \ref{IIv}, Parts (b) and (d).
Indeed,  one has in the normal form   at least $k-t+1$ $\overline{a}$'s  before  index $k+q$.
Since  $[\sigma^n][\ell]=a$, this  gives $k+q-\ell \geq k-t+1$,  i.e., $q \geq \ell-t+1$.

Similarly, the second part follows from Lemma \ref{lem:propa}, Assertion \ref{IIv}, Parts (a)
and (c).
\end{proof}

\subsection{More on the second level of block-transformations}\label{subsec:moreIII}

The following lemma provides a characterization of  the  case
where a second level of block-transformations has to be performed.
 We recall that the supports of the propagated   errors  of the first level in
$Z_{\sigma^n}$  are $[i+kj,  p+k(j+1)]$ for $j=0, \dots, n-2$.


\begin{lemma}\label{lem:type3pre}
Let $\sigma$ be an episturmian substitution of positive length $k$
which admits at least two distinct letters in its normal form, and
such that $Z_{\sigma^n}$ contains a simple error, with $ n\geq
2$.  Let $a$ be the letter of  the leftmost simple error in
$Z_{\sigma^n}$. The word $Z'_{\sigma^n}$ obtained after the first
level of
block-transformations 
is not in normal form if and only if the next two conditions hold:

\begin{enumerate}
\item  \label{III:i} $\theta_{\sigma}(a)=a$ and $ n \geq 3$;

\item \label{III:ii}  between two  consecutive   propagated   errors  of the first level in
$Z_{\sigma^n}$, one has
$Z_{\sigma^n}[p+kj+1,i+kj-1]=\overline{b_1}\cdots\overline{b_t}$,
for  all  $j=1, \dots, n-2$  (with $t$ possibly equal to $0$),  where the
letters $b_{\ell}$'s  belong to ${\mathcal A}$ and are distinct
from $ a$ for all $\ell$. 
\end{enumerate}
\end{lemma}
\begin{proof}
Suppose that   there is  a simple  error  in   $Z'_{\sigma^n}$
once we have finished the first level of block-transformations.
 Let $[\ell, \ell'] $ be the support   of the leftmost occurrence  of a simple  error  in $Z'_{\sigma^n}$  and let $b \in {\mathcal A}$ stand for its letter.

 Let us prove by contradiction that $b=a$.  Let us  assume $ b \neq a$.  
 By $k$-periodicity,   $[\ell, \ell'] $ intersects the  support  $[i,k+p]$ of the propagated $a$-error of the first level we just normalized during the first level of block-transformations.
The support of the error cannot contain the position $i$ since we
have $Z'_{\sigma^n}[i]=a$ by Lemma \ref{lem:propa}, Assertion
\ref{IIii}, and a simple error of letter $b$ cannot contain any  other
letter without spin. So, we have $\ell>i$. However, we cannot have
$\ell\in [i,k+p]$, since a simple error of letter $b$ starts with
$\overline{b}$, whereas there are no occurrences of
$\overline{b}$
 in $Z'_{\sigma^n}[i,k+p]$  by Lemma \ref{lem:propa}, Assertion \ref{IIi}, since we assumed $ b \neq a$.  We thus get a contradiction, which implies that  $b=a$.

So, the new simple error is of the form $\overline{a}
\overline{b_1}\dots \overline{b_{t}}a$ with $b_i\in
\mathcal{A}\setminus \{a\}$  and $t\geq0$,
and  it is not contained completely in the support $[i,k+p]$
 of the leftmost error of the first level. Hence
this simple error starts at index $k+p$ (the last occurrence of
$\tilde{a}$ in the propagated $a$-error of the first level), i.e.,
$\ell=p+k$. Moreover, the occurrence of $a$ in it must be at index
$k+i$, i.e., $\ell' =k+i$ (so that
$Z_{\sigma^n}[k+i]=\overline{a}$ and $Z'_{\sigma^n}[k+i]=a$).
Otherwise the factor $Z_{\sigma^n}[i,k+p]$ could be continued to
the right keeping the form \eqref{eq:max_factor} and hence would
not be maximal, and  the  propagated error of the first level
would  not actually end at the index $k+p$. This implies that
$\theta_{\sigma}(a)=a$ and $ n\geq 3$.
We  then use the   $k$-periodicity for  the proof of Assertion
\ref{III:ii}.

The converse is straightforward: indeed, if we have in
$Z_{\sigma^n}[p+kj,i+kj]=a\overline{b_1}\cdots\overline{b_t}\overline{a}$,
for $j=1, \dots, n-2$, then
$Z'_{\sigma^n}[p+kj+1,i+kj-1]=\overline{a}\overline{b_1}\cdots\overline{b_t}a$,
which is a simple error.
\end{proof}

This proves in particular that all  propagated  errors of the first
level  have the same type. In particular, after the first level of
block-transformations, if a new simple error is created in the
$j$-th $k$-period,  with $j=2, \dots, n-1$, then
 new simple errors are created in each $k$-period of index $j$,  with  $j=2, \dots, n-1$ (all periods except for the first one and the last one).

In particular, Lemma \ref{lem:type3pre} implies the following:

\begin{corollary} \label{IIsuppbis}  Let $\sigma$ be an episturmian substitution which admits at least two distinct letters in its normal form.
 Let $n \geq 2$. Let $a$ be the letter of  the leftmost simple error in $Z_{\sigma^n}$.
 If  the errors are of type I and $\theta_{\sigma}(a)=a$, then
$ p < i-1$, and  there  exists an occurrence of a letter $b \in
{\mathcal A}$ with  $b \neq a$   in $[k+p+1, k+i-1]$ for each $k=0,\ldots, n-1$.
\end{corollary}

%
%

Lemma  \ref{lem:type3}  below describes a  word $Z_{\sigma^n}$ containing
errors of type II  (see Definition \ref{def:typeII}) and the  action of the second  level of
block-transformations on  $Z'_{\sigma^n}$.    According to Lemma \ref{lem:type3pre},   $ n \geq 3$, $\theta_{\sigma}(a)=a$,  i.e., all simple errors  to be normalized have the same letter  $a$.
Moreover all errors are occurrences of non-normal spinned words
of the form $\overline{a }\overline{b_1} \cdots \overline{b_t} a$,
with the $b_i$'s  being distinct from $a$ and $t$ possibly equal
to $0$: between two consecutive errors of the first level in
$Z'_{\sigma^n}$ (at indices  $[p+1+jk,i-1+jk]$, with $j=1,\ldots
n-2$),
one has  a factor of  the form $\overline{b_1} \cdots
\overline{b_t}$  with  the  letters $b_j$'s being
distinct from $ a$ for all $j$ (here we also use Lemma
\ref{lem:propa}, Assertions (\ref{IIii}) and (\ref{IIiii})).
Whereas the supports of
the errors of type I are disjoint, we see below that  in type II  the
leftmost propagated  error  of the first level that occurs in $Z_{\sigma^n}$
 propagates    on all the  $k$-periods of  $Z_{\sigma^n}$, i.e., on  $[i, p +k(n-1)]$ (with Notation \ref{notation}).

\begin{lemma}\label{lem:type3} Let $\sigma$ be an episturmian substitution of positive length which admits at least two distinct letters in its normal
form and 
assume that the errors in $Z_{\sigma^n}$ are of  type $II$. Let
$a$ be the letter of the leftmost error in $Z_{\sigma^n}$. Let
$Z'_{\sigma^n}$ be  the word obtained after the first level of
block-transformations (as described in Lemma \ref{lem:propa}).
Then the following holds.

\begin{enumerate}
 \item \label{III:iii}  Any occurrence of a  letter  $\tilde{b}$, with $b\neq a$,   has a  spin   equal to $1$ in
$Z_{\sigma^n}$.
\item \label{III:iv}   Inside   the  first  period  of $Z_{\sigma^n}$,   each occurrence of $\tilde{a}$   has spin $0$ at indices  smaller than or equal to $p$ and  spin $1$  after that.


\item \label{III:vi} The second level of block-transformations is
performed in at most two steps.

\begin{itemize} \item First, we correct the simple errors $\overline{a}\overline{b_1}\dots\overline{b_t}a$ in
$Z'_{\sigma^n}$ at indices $p+kj, \ldots, i+kj$, for $j=1, \dots,
n-2$. The factors of the form $\overline{b_1}\dots\overline{b_t}$
in $Z'_{\sigma^n}$ at indices $p+kj+1, \ldots, i+kj-1$, for $j=1,
\dots, n-2$ become ${b_1}\dots{b_t}$.

\item Secondly,  for
$j=1, \dots, n-2$, we consider the maximal factors
of the form $\tilde{a}^s$ with support $[p-s+1+kj, p+kj]$(maximal in the sense that $Z'[p-s]\neq
\tilde{a})$ and $\tilde{a}^q$  with support $[i+kj, i+q-1+kj]$
(so that $Z'[i+q]\neq \tilde{a}$). They are normalized using
block-transformations $\overline{a}a\to a \overline{a}$.

More precisely, the following happens.

 \begin{itemize} \item  If $s=1$, then in the normal form $[\sigma^n]$ we have $[\sigma^n][p+kj]=a$,  for $j=1,\dots n-2$.
Symmetrically, if $q=1$, then
$[\sigma^n][i+kj]=\overline{a}$  for $j=1,\dots n-2$. If $s=1$ and
$q=1$, then the second level of block-transformations is over.

\item If $s>1$ or $q>1$, then we continue the  block-normalization
as follows.
\begin{itemize}
\item[$\bullet$]

If there is a letter $c$ with $c \neq \tilde{a}$ in $Z'_{\sigma^n}[i, k+p]$, then the intervals
$[p-s+1+kj+k, p+kj+k]$ and $[i+kj, i+q-1+kj]$ do not intersect, for
$j=1, \dots, n-2$.
If $s>1$ and  $Z'_{\sigma^n}[ p-s+1+kj,p+kj]=a^{\ell}
\overline{a}^{s-\ell}$ for $j=1, \dots, n-2$, then  $[\sigma^n][ p-s+1+kj,p+kj]=a^{\ell+1} \overline{a}^{s-\ell-1}$ ($s>\ell$ due
to Assertion (\ref{IIiii}) from  Lemma \ref{lem:propa}). Symmetrically, if $q>1$ and  $Z'_{\sigma^n}[i+kj, i+q-1+kj]=a^{r} \overline{a}^{q-r}$ for $j=1, \dots,
n-2$, then  $[\sigma^n][i+kj, i+q-1+kj]= a^{r-1}
\overline{a}^{q-r+1}$ ($r>0$ due to Lemma \ref{lem:propa}).

 \item[$\bullet$]  If  $Z'_{\sigma^n}[i, k+p]$ contains only occurrences of $\tilde{a}$ (there is thus a central part),  then
$Z'_{\sigma^n} [i+kj, p+kj+k] =a^{\ell} \overline{a}^{k+p-i-\ell+1}$ ($j=1, \dots, n-2$); moreover,
these positions are not changed for $j=1, \dots n-3$,  i.e.,
$[\sigma^n][i+kj, p+kj+k] =a^{\ell} \overline{a}^{k+p-i-\ell+1}$;    $[\sigma^n][i, p+k] =a^{\ell +1} \overline{a}^{k+p-i-\ell}$ ($j=0$)
and $[\sigma^n][i+k(n-3), p+k(n-2)] =a^{\ell-1} \overline{a}^{k+p-i-\ell+2}$ ($j=n-3$),  if  $n \geq 4$.
\end{itemize}

\end{itemize}
\end{itemize}
After these transformations, we get the normal form of $\sigma^n$.


\item  \label{III:vii}   In $[\sigma^n]$, any occurrence of a  letter  $\tilde{b}$, with $b\neq a$,
occurs with spin    $0$
 in $[i,p+k(n-1)]$,  and  with  spin $1$  in  both the prefix of length $i-1$ and  the suffix of length ${k-p}$.

\item \label{III:viii}    One  has $[{\sigma^n}][i]=a$,   $[\sigma^n][k(n-1)+p]=\overline{a}$;

\item  \label{III:ix}  The spin   of a letter $\tilde{a}$  at any index
which neither  belongs to the central  part (if there is a  central   part), nor   equals  $i$, nor equals  $p+(n-1)k$,
  is the same in $Z_{\sigma^n}$ and  in $[\sigma^n]$.

\end{enumerate}
\end{lemma}

\begin{proof}
Assertion (\ref{III:iii}) follows from Lemma \ref{lem:type3pre},
Assertion (\ref{III:ii}) (for indices between $p+kj$ and $i+kj$),
and Lemma \ref{lem:propa}, Assertion (\ref{IIi}) (for indices
between $i+kj$ and $p+kj$).

Assertion (\ref{III:iv}) follows from Lemma \ref{lem:propa},
Assertion (\ref{IIbis}) and the fact that there are no occurrences
of $\tilde{a}$ between indices $p$ and $i$ when we have errors of
type II (Lemma \ref{lem:type3pre}, Assertion (\ref{III:i})).

The normalization process described in Assertion (\ref{III:vi}) is
a direct result of block-transformations applied first to the
simple errors $\overline{a}\overline{b_1}\dots\overline{b_t}a$,
and then to the neighbouring factors of the form $\tilde{a}^+$ (if
they are of length greater than 1). So, for Assertion
(\ref{III:vi}) we only need to prove that after these
block-transformations the obtained word is in normal form, i.e.,
does not contain any errors.

Now we proceed as follows. We first prove that the claims of
Assertions (\ref{III:vii})--(\ref{III:ix}) hold for the word
obtained after the second level of block-transformations described
in Assertion (\ref{III:vi}). Then we prove Assertion
(\ref{III:vi}) (that there are no errors after these
transformations). As a corollary, we will have that after the second
level of block-transformations we get the normal form, so the
Assertions (\ref{III:vii})--(\ref{III:ix}) hold true for normal
forms as stated.

Assertion (\ref{III:vii}) follows from Assertion (\ref{III:vi})
for indices between $p+kj$ and $i+kj$, $j=1, \dots, n-2$, and from
Lemma \ref{lem:propa}, Assertion (\ref{IIi}) for indices in the
support of the error of the first level (i.e., between $i+kj$ and
$p+kj+k$, $j=0, \dots, n-2$). The letters at indices before $i$
and after $k(n-1)+p$ are not touched by the first level  block-transformations by definition and also by the normalization described
in Assertion (\ref{III:vi}), so they stay unchanged (and hence
have spin 1 by Assertion (\ref{III:iii})).

Assertion (\ref{III:viii}) follows from Lemma \ref{lem:propa},
Assertions (\ref{IIii}) and (\ref{IIiii}) and $k$-periodicity,
since these indices are not touched by the block-transformations
described in Assertion (\ref{III:vi}).

Assertion (\ref{III:ix}) follows Assertion (\ref{III:vi}) and from
Lemma \ref{lem:propa}, Assertions (\ref{IIiv}) and (\ref{IIv}).

Now we prove that there are no more errors in the obtained word (Assertion (\ref{III:vi})). Suppose there are errors after the two
levels of block-transformations.

First note that the errors must be of the letter $a$. Indeed,
errors of another letter $b\neq a$ must be of the form
$\overline{b} \overline {c_1} \cdots \overline {c_\ell} b$ (say at
index $j$). Due to Assertion (\ref{III:vii}), we must have either
$j < i$ or $j> p+(n-1) k$, since only at these indices we can have
$\overline{b}$ for $b\neq a$. The second case is impossible due to
absence of letters $b$ of spin 0 after $ p+(n-1) k$, and the first
case is impossible since the first letter of spin 0 in the
conjectured simple error must be $b$ while it is $a$ at index $i$.
So, we indeed cannot have an error of a letter $b\neq a$.

Now we are going to prove that we cannot have errors of the letter
$a$ either. First note that there are no errors $\overline{a}
\overline {c_1} \cdots \overline {c_\ell} a$ for $\ell>0$: indeed, by
the previous assertions, series of occurrences of letters
$\overline {c}$ for $c\neq a$ can only occur in the first period
before index $i$ and in the last period after index $k(n-1)+p$.

If such a series occurs in the first period, this means that the
simple error $\overline{a} \overline {c_1} \cdots \overline {c_\ell}
a$ either ends at the position $i$ or is entirely contained in the
prefix of length $(i-1)$ of the word. The first case is impossible.
Indeed, by  Lemma \ref{lem:type3pre}, the occurrence of
$\tilde{a}$ preceding the position $i$ is at position $p$ and it
has spin 0 by Assertion (\ref{III:iii}); moreover, due to the fact that
this position is not touched by the first and second levels of
transformations,   its spin must be 1 to have an error. The second
case is impossible since we did not touch the prefix of length
$(i-1)$ on the first and second level of block-transformations, so
it coincides with the prefix of $[\sigma]$ which is in normal form
and hence cannot contain errors. For the last period the proof is
symmetric.

It remains to prove that we cannot have errors $\overline{a}a$.
Note that we cannot have errors inside the parts we normalized
during the first level of normalization and did not change during
the second level of normalization (between
 $i+q-1+kj$ and $p-s+k(j+1)$), as well as in the parts we normalized during the second level of
 normalization  described in Assertion (\ref{III:vi})
(between $p+q-1+kj$ and $i-s+kj$) or inside prefixes or suffixes
of $[\sigma]$ (before $i$ and after $p+k(n-1)$). So, it remains to
show that we cannot have errors  $\overline{a}a$ on the boundary
between these normalized parts, i.e., at positions $p-s+kj$ and
$i+q-1+kj$. This follows from the fact that we defined $s$ and $q$
such that $Z'_{\sigma^n}[p+kj-s]\neq \tilde{a}$, and
$Z'_{\sigma^n}[i+kj+q]\neq \tilde{a}$.
We remark
that we use the fact that we have at least two letters in the
normal form (also when we define the maximal intervals of the form
$\tilde{a}$: when there are at least two letters, then their
lengths are smaller than $k$, otherwise if there is only one
letter, they appear to be one interval covering the whole word).

 So, we
proved that after the second level of block-transformations we
have normalized forms, hence Assertions
(\ref{III:vii})--(\ref{III:ix}) hold for normal forms.

Figure \ref{fig:level3} corresponds to the case where  there is no
central part.

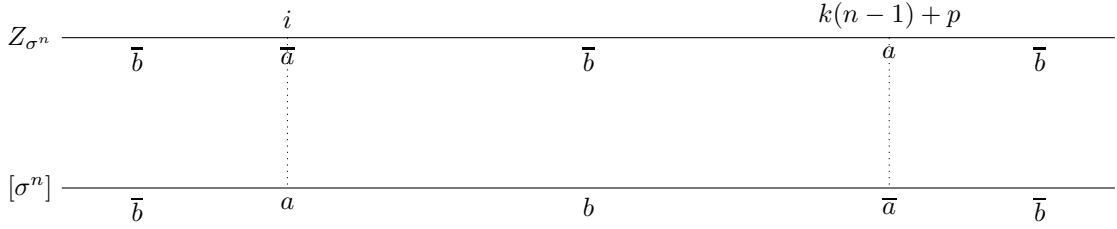
\begin{figure}
\begin{tikzpicture}
\draw  (1,0) -- (15,0);
\draw (4,0) node [above] {$i$};
\draw (4,0) node [below] {$\overline{a}$};
\draw (4,-2) node [below] {$a$};
\draw (2,-2) node [below] {$\overline{b}$};
\draw (14,-2) node [below] {$\overline{b}$};
\draw (14,0) node [below] {$\overline{b}$};

\draw (2,0) node [below] {$\overline{b}$};

\draw (8,0) node [below] {$\overline{b}$};

\draw (8,-2) node [below] {$b$};

\draw (12,-2) node [below] {$\overline{a}$};

\draw (12,0) node [above] {$k(n-1)+p$};
\draw (12,0) node [below] {$a$};
\draw [dotted] (4,0)-- (4,-2);
\draw [dotted] (12,0)-- (12,-2);

\draw (1,0) node[left]{$Z_{\sigma^n}$};
\draw (1,-2) node[left]{$[\sigma^n]$};
\draw (1,-2) -- (15,-2);
\end{tikzpicture}
\caption{Illustration   of the transformations  performed  during  the two levels of block-transformations, in  the case where  there  is no central part.} \label{fig:level3}
\end{figure}\end{proof}

The following examples  illustrate Lemma \ref{lem:type3}.
\begin{example}
{\bf{(Error of type II: illustration of Assertion (\ref{III:vi}) from Lemma
\ref{lem:type3}, with a letter
$d$ with $d \neq \tilde{a}$ in $Z'_{\sigma^n}[i, k+p]$)}}

Let $\sigma$ be a
substitution with normal form
$$[\sigma]=a^3 \overline{b c} \overline{a} \overline{d}$$ with $a,b,c,d$
distinct letters and $\theta_{\sigma}$ being the identity. The
letters of the supports  of the   propagated errors of the first
level errors are marked by red;
letters in  blue indicate letters that are involved in the
second level of  block-transformations. One has
$$Z_{\sigma^4 }= a^3 \overline{bc} {\color{red}
\overline{a} \overline{d} | a^3 } {\color{blue}\overline{b}
\overline{c}} {\color{red} \overline{a} \overline{d} | a^3 }
\overline{bc} {\color{red}\overline{a} \overline{d} | a^3 }
\overline{bc} \overline{a} \overline{d}.$$

The  first level of  block-transformations (which deals with the
red part) yields:

$$Z'_{\sigma^4}= a^3 \overline{b c}
{\color{red}   a  d | a^2 \overline{a}}
{\color{blue}\overline{bc}} {\color{red}
 a  d | a^2 \overline{a}}
{\color{blue}\overline{bc}} {\color{red}
 a d | a ^2\overline{a} } \overline{bc} \overline{a} \overline{d}.$$

And we create another simple error $\overline{a}\overline{bc}a$ (involved letters are marked by  red below), and fixing it
we get

$$ a^3 \overline{b c}
  a  d | a^2 {\color{red} a
bc \overline{a}}
   d | a^2 {\color{red} a
bc \overline{a}} d | a^2 \overline{a}  \overline{bc} \overline{a}
\overline{d}.$$

 In this example there are no errors in the series of $\tilde{a}$'s, so
 we obtainded the normal form on the previous step:

$$[\sigma^4]= a^3 \overline{b c}
  a  d | a^3
bc \overline{a}
   d | a^3
bc \overline{a} d | a^2 \overline{a}  \overline{bc} \overline{a}
\overline{d}.$$

\end{example}

\begin{example} {\bf (Error of type II: illustration of Assertion (\ref{III:vi}) from Lemma
\ref{lem:type3},  when $Z'_{\sigma^n}[i, k+p]$ contains only occurrences of $\tilde{a}$)}\label{ex:new}

We  consider
$\sigma$ with normal form
$$[\sigma]=a^3 \overline{b c} \overline{a}^3$$ with $a,b,c$
distinct letters and $\theta_{\sigma}$ being the identity. The
letters of the supports  of the   propagated errors of the first
level errors are first  marked by red;  this  coincides with the  central
part; letters in  blue indicate letters that are involved in the
second level of  block-transformations. One has
$$Z_{\sigma^4 }= a^3 \overline{bc} {\color{red}
\overline{a}^3 |  a^3} {\color{blue}\overline{bc}} {\color{red}
\overline{a}^3 | a^3 } {\color{blue}\overline{b} \overline{c}}
{\color{red} \overline{a}^3 | a^3 } \overline{bc}
\overline{a}^3.$$

The  first level of  block-transformations (which deals with the
red part) yields:

$$Z'_{\sigma^4}= a^3 \overline{b c}
{\color{red}  a^3 |\overline{a}^3} {\color{blue}\overline{bc}}
{\color{red}
 a^3 |\overline{a}^3} {\color{blue}\overline{bc}}
{\color{red}
 a^3 | \overline{a}^3 } \overline{bc} \overline{a}^3.$$

And we create  as before  simple errors  $\overline{a}\overline{bc}a$;
fixing them as in the previous example we get:

$$a^3 \overline{bc} a^3 | \overline{a}^2 { \color{red}a bc\overline{a}}a^2 | \overline{a}^2 {\color{red}a bc \overline{a}} a^2 |\overline{a}^3 \overline{bc} \overline{a}^3  .$$

So here it remains to correct the simple  errors $\overline{a}a$ in the series
of $\tilde{a}$ (their supports are depicted in red below), which yields
$$  a^3 \overline{bc}
{\color{red} a^3 |
 a \overline{a}^2 } bc {\color{red} a^2 a
 | \overline{a}
  \overline{a}^2 } bc {\color{red}
a^ 2 \overline{a} | \overline{a}^3} \overline{bc} \overline{a}^3,$$
and thus

$$[\sigma^4]=
a^3 \overline{bc}
 a^3 |
 a \overline{a}^2  bc a^3
 | \overline{a}^3  bc
a^ 2 \overline{a} | \overline{a}^3\overline{bc} \overline{a}^3  .$$
\end{example}

\begin{remark}\label{rem:new}
\begin{enumerate}
\item
We  assume  that $Z_{\sigma^n}$ admits a simple error.
 Let $i$ stand  for the index of the leftmost   error  in $Z_{\sigma^n}$. The  prefix of length $i-1$  of the normal form $[\sigma]$  is a  prefix of $Z_{\sigma^n}$.
\item \label{rem:3}
 In the  full process of block-normalization we only make
block-transformations of the form $\overline{a b_1} \cdots
\overline{{b_s}} a \to a {b_1} \cdots {b_s} \overline{a}$.  Hence,
once the spin of a letter $b \neq a$ is  equal to $0$ in the  full
normalization process at some index  in the support of this error,
it  will never change back to spin $1$ at this index.
\end{enumerate}
\end{remark}

\subsection{More about normal forms}\label{subsec:morenormalform}

In the proof of Lemma \ref{lem:mainlem} we will use the following   technical general proposition about normal
forms.
It is used in particular in the proof of Lemma  \ref{lem:type32}.
\begin{proposition}\label{prop:prefix}  Let $\sigma$ and $\tau $ be episturmian  substitutions. We assume that $\sigma^n=\tau^m$ for $n,m$ positive integers.
We assume that  $Z_{\sigma^n}$ and $Z_{\tau^m}$ both contain
errors and  that both leftmost errors  are of the same letter $a$.
We also  assume that  all the  errors are of letter $a$  in
$Z_{\sigma^n}$ and $Z_{\tau^m}$, i.e., $\theta_{\sigma}(a)=a$ and
$\theta_{\tau}(a)=a$.

\begin{enumerate}
\item \label{prop:prefixi}
If one has
$|\pref_\ell(Z_{\sigma^n })|_{a}=|\pref_\ell(Z_{\tau^m})|_{a}$ for some  positive integer $\ell$,
then the substitutions  $\mu(\pref_\ell(Z_{\sigma^n }))$ and $ \mu (\pref_\ell(Z_{\tau^m}))$
are equal.
\item  \label{prop:prefixii} If,  for some  index $\ell$,  $Z_{\sigma^n}
[\ell]=Z_{\tau^m}[\ell]=\tilde{b}$ for some  letter $b\neq a$,
then $\mu  (pref_\ell(Z_{\sigma^n}))=\mu (
\pref_{\ell}(Z_{\tau^m}))$ and also,    by left-cancellativity,
$\mu (pref_{\ell-1}(Z_{\sigma^n}))=\mu (
\pref_{\ell-1}(Z_{\tau^m}))$. In particular, we have
$|\pref_\ell(Z_{\sigma^n})|_{a}=|\pref_{\ell}(Z_{\tau^m})|_{a}$.

\item  \label{prop:prefixiii} If, for some index $\ell$,
 $Z_{\sigma^n} [\ell]=b$ and $Z_{\tau^m}[\ell]=\overline{b}$
for some  letter $b\neq a$, then
$|\pref_\ell(Z_{\sigma^n})|_{a}=|\pref_{\ell}(Z_{\tau^m})|_{a}+1$
and $|\pref_{\ell}
(Z_{\sigma^n})|_{\overline{a}}=|\pref_{\ell}(Z_{\tau^m})|_{\overline{a}}-1$.
\end{enumerate}
\end{proposition}

\begin{proof}  First observe that $[\sigma^n]=[\tau^m]$.
The proof  of  Assertion (\ref{prop:prefixi})  is  by induction on  $\ell$. Suppose that for each
$\ell'<\ell$ the statement holds.
The induction property  clearly  holds  for $\ell=1$.

Note first that  $|\pref_\ell(Z_{\sigma^n })|_{a}=|\pref_\ell(Z_{\tau^m})|_{a}$  implies that
$|\pref_\ell(Z_{\sigma^n })|_{\overline{a}}=|\pref_\ell(Z_{\tau^m})|_{\overline{a}}$, by  Assertion (\ref{rem:4}) in Remark \ref{remarksbis}.

We first assume  that  at index $\ell$ in $Z_{\sigma^n}$, the
letter is    some  $\tilde{b}$ with $b\neq a$. At index $\ell$ in
$Z_{\tau^m}$, one has  the same letter $\tilde{b}$ (possibly with
a different spin),
 by  Assertion (\ref{rem:4}) in Remark \ref{remarksbis}. If  the  occurrences  at index $\ell$
of the letter $b$
in $Z_{\sigma^n}[\ell]$
and $Z_{\tau^m}[\ell]$ have the same spin, the proof follows from the
fact that the episturmian monoid is left-cancellative  (see Proposition \ref{prop:cancel}),   by considering   the prefix
of length $\ell-1$  together with the induction hypothesis for $\ell-1$.
And they cannot have different spins, since otherwise during the
normalization of the power of the substitution having $\overline{b}$ we must apply a
rule involving this position, which changes the number of $a$'s
in the prefix, and we can do it only once by Assertion (\ref{rem:3}) of Remark  \ref{rem:new}. Since $[\sigma^n]=[\tau^m]$, this is not possible.

We now assume that  $Z_{\sigma^n}[\ell]=\tilde{a}$ (and hence
$Z_{\tau^m}[\ell]$ is equal to $a$ or $\overline{a}$  by the above). As above, if
$Z_{\sigma^n}[\ell]$ and $Z_{\tau^m}[\ell]$ have the same spin, the proof
follows  by
left-cancellativity and from the induction hypothesis for $\ell'=\ell-1$.
So, we   have to  consider the case where in one of the representations
we have $a$, and in the other one $\overline{a}$.
Without loss of generality, assume that $Z_{\sigma^n}[\ell]={a}$ and
$Z_{\tau^m}[\ell]=\overline{a}$.

\begin{itemize}
\item  If, for some $\ell' < \ell$,
$|\pref_{\ell'}(Z_{\sigma^n})|_{a}=|\pref_{\ell'}(Z_{\tau^m})|_{a}$
and
$|\pref_{\ell'}(Z_{\sigma^b})|_{\overline{a}}=|\pref_{\ell'}(Z_{\tau^m})|_{\overline{a}}$,
then we are done by induction and left-cancellativity.

\item So, it
remains to consider the case where, for each length $\ell' < \ell$, we have
$|\pref_{\ell'}(Z_{\sigma^n})|_{a} \neq |\pref_{\ell'}(Z_{\tau'})|_{a}$.
Since $Z_{\sigma^n}[\ell]={a}$ and
$Z_{\tau^m}[\ell]=\overline{a}$, then
$|\pref_{\ell'}(Z_{\sigma^n})|_{a}<|\pref_{\ell'}(Z_{\tau'})|_{a}$
and
$|\pref_{\ell'}(Z_{\sigma^n})|_{\overline{a}}>|\pref_{\ell'}(Z_{\tau^m})|_{\overline{a}}$, for all $ 1 \leq  \ell' < \ell$.

In particular, for $\ell'=1$ this means that
$Z_{\sigma^n}(1)=\overline{a}$ and $Z_{\tau^m}(1)={a}$.
Since $[\sigma ^n]= [\tau^m]$  and the only modifications
during normalization are of the form $\overline{a b_1}
 \cdots \overline{{b_s}} a \to
a {b_1} \cdots {b_s} \overline{a}$ (see Assertion (\ref{rem:3}) of Remark \ref{rem:new}), we have
$[\sigma^n][1]=[\tau^m][1]=a$. Now, to convert the
first symbol $\overline{a}$ in $Z_{\sigma^n}$ to $a$ during  the  full process of  block-normalization from $Z_{\sigma^n}$ to $[\sigma^n]$,
we need to have a prefix of the form $\overline{a} \{
\overline{c} | c\in {\mathcal A}\}^* a$ of length at most
$\ell$ ($c \neq a$).  Indeed, consider  the first occurrence of $a$ in $Z_{\sigma^n}$. Before it,   we
cannot have an occurrence of a letter  $b$ with $b \neq a$ with spin  $0$, since it  would  block the
prefix for the  modifications that have to be done   (by Assertion (\ref{rem:3}) of Remark  \ref{rem:new}),  contradicting  the existence of  $a$ as a first letter for
$[\sigma^n]$.   Hence  there is a propagated   $a$-error of the first  level  starting at  index $1$ in $Z_{\sigma^n}$ and   with support included in $[1, \ell]$.
By normalizing  this error,   we obtain that the   first letter of  the word $Z'_{\sigma^n}$ is $a$,  and
again by left-cancellativity  applied to  $\mu(Z'_{\sigma^n} )=\mu(Z_{\tau^m})$ for the prefix of length $1$ (which is $a$) and by
induction hypothesis for $\ell'=\ell-1$. This ends the proof of Assertion  (\ref{prop:prefixi}). \end{itemize}

We now prove  Assertion (\ref{prop:prefixii}). Assume first that
both $Z_{\tau^m}$ and $Z_{\sigma^n}$ have $b$ at  position
$\ell$.  The normalization  is conducted independently  before and
after  index $\ell$. Assume now  that both $Z_{\tau^m}$ and
$Z_{\sigma^n}$ have $\overline{b}$ at  position $\ell$. Then,
during the normalization, it can only change to $b$ once and in both
$Z_{\tau^m}$ and $Z_{\sigma^n}$. As a  result of application of
the rule which involves $\ell$, the number of $a$'s in the prefix
increases by $1$, and no other rules change the number of $a$'s in
this prefix. So, by Assertion \ref{prop:prefixi}, the prefixes
correspond to equal substitutions.

For the proof of the second part of Assertion
(\ref{prop:prefixii}) and Assertion (\ref{prop:prefixiii}), it is
enough to notice that the only rule changing the number of $a$'s
and $\overline{a}$'s in the prefix of length $\ell$ is the one
changing the spin of the occurrence of $\overline{b}$ at the
position $\ell$ to $b$, which changes the numbers of $a$'s and
$\overline{a}$'s in the prefix by $1$. If the
conditions of Assertion (\ref{prop:prefixiii})  hold,  this rule is applied once for $Z_{\tau^m}$ during the
normalization, and if the
conditions of Assertion (\ref{prop:prefixii})  hold, it is either applied
for both $Z_{\sigma^n}$ and $Z_{\tau^m}$ once, or it is never
applied to either of them.
\end{proof}

The following proposition  is a generalization of Proposition \ref{prop:prefix} when  the errors  inside  $Z_{\sigma^n}$ or $Z_{\tau^m}$  do not have possibly  the same letter.

\begin{proposition}\label{prop:prefix2}
 We assume that  $Z_{\sigma^n}$ and $Z_{\tau^m}$ both contain  errors and  that both leftmost errors  have  the same letter  $a$.
 Let $[i,k+p]$ and $[i',k'+p']$
stand for the respective supports of the leftmost errors
($|\sigma|=k$ and $|\tau|=k'$).
 All the assertions of Proposition  \ref{prop:prefix}  also hold if
 \begin{enumerate}
 \item \label{prefix2:i} $\theta_{\sigma}(a)\neq a$ and  $\theta_{\tau}(a) = a$ and
 $\ell <k+i$,
 \item  \label{prefix2:ii} or (symmetrically) if  $\theta_{\tau}(a)\neq a$ and $ \theta_{\sigma}(a) = a$  and $\ell<k'+i'$,
 \item   \label{prefix2:iii} or (both) if
   $\Theta_{\sigma}(a)\neq a$ and $ \theta_{\tau}(a)\neq a$ and $\ell <\min(k+i, k'+i')$.
   \end{enumerate}
 \end{proposition}

 \begin{proof} The proof basically repeats the proof of Proposition \ref{prop:prefix}  (and  Remark  \ref{rem:new}   also holds for the index $\ell $ in the indicated prefix).
\end{proof}

\subsection{The case of a one-letter normal word} \label{subsec:oneletter}
In this section, we handle  the normalization of $[\sigma]^n$
via block-trans\-for\-ma\-tions in the case where the normalized
directive word of the episturmian substitution $\sigma$ contains
only one letter $\tilde{a}$, i.e., all its letters are equal to
$\tilde{a}$ (this will be used in Lemma \ref{lemma:oneletter}). We
thus have for the normal form of $\sigma$:
$$[\sigma]=a^{s}\overline{a}^t\theta_{\sigma}$$ for some integers
$s,t\geq 0$,  with $s+t\geq 1$.

Consider first  the case where $\theta_{\sigma}(a)=a$. After
applying several times the block-transformation $\overline{a}a\to
a\overline{a}$, we get   for the normal form of $\sigma^n$, for $n
\geq 1$:
$$[\sigma^n]=a^{sn}\overline{a}^{tn}\theta_{\sigma}^n.$$

Now consider the case where $\theta_{\sigma}(a)\neq a$. To get the
normal form of $[\sigma^n]$ from $[\sigma]^n$, we only have to
move the permutation to the right, using the rule
\eqref{eq:permut}. So,
$$[\sigma^n]=a^{s}\overline{a}^{t} \mid
\theta_{\sigma}(a)^{s}\overline{\theta_{\sigma}(a)}^{t} \mid
(\theta^2_{\sigma}(a))^{s}\overline{\theta^2_{\sigma}(a)}^{t} \mid
\cdots \mid
(\theta^{n-1}_{\sigma}(a))^{s}\overline{\theta^{n-1}_{\sigma}(a)}^{t}
\mid \theta_{\sigma}^n.$$

\section{Decomposition lemmas and proof of  Lemma \ref{lem:mainlem}.}\label{sec:lemma}

The aim of this section is to prove Lemma \ref{lem:mainlem}.   We assume that $\sigma^n =\tau^m$, with $ n \geq m\geq 1$.
Observe that
$$[\sigma^n ]=[\tau^m], \qquad Z_{\sigma^n} \equiv [\sigma^n] = [\tau^m]
\equiv Z_{\tau^m}.$$ We want to prove that there exists an
episturmian substitution  $\varrho$   such that  $\tau =\sigma \circ
\varrho$. The letter $k$ stands for  the length $|\sigma|$  of
$\sigma$ and $k'$  stands for  the length  $| \tau | $ of  $\tau$.
One has
$$kn= k'm, \qquad k' \geq k.$$

 We use Notation \ref{notation}. Moreover, when working with $\sigma $ and $\tau$,
 we use letters with primes for $\tau$.  For instance, in the case where there are errors in  both $Z_{\sigma^n}$ and
$Z_{\tau^m}$,
  $a$    stands for the  letter of the  leftmost error in  $Z_{\sigma^n}$ and   $a'$    stands for the  letter of the  leftmost error in $Z_{\tau^m}$.

The proof of  Lemma \ref{lem:mainlem} is given in Section \ref{proof:mainlem} and relies on a succession of lemmas  stated and proved  in Section  \ref{subsec:lemmas}.
The following remark will   be used at several places in the proof.

\begin{remark} \label{rem:prefix}
We recall from Section \ref{subsec:normalform}
that
$[\sigma]= w_{\sigma} \theta_{\sigma}$ , where $w_{\sigma}$ is the normalized directive  word of $\sigma$.
If  the  normalized  directive word $w_{\sigma}$  of  $\sigma$  is a prefix of  the  normalized   directive word  $w_{\tau}$  of  $\tau$,
then  there exists an  episturmian substitution $\varrho$ such that  $\tau=\sigma\circ \varrho$. Indeed, let $w'$   be such that $w_{\tau}= w_{\sigma}w'$. One has
 $$\varrho = \theta_{\sigma}^{-1} \circ w' \circ \theta_{\tau}.$$
\end{remark}

\subsection{Proof of  Lemma \ref{lem:mainlem}} \label{proof:mainlem}

We first can   restrict ourselves to the case where  $m \geq 2$.
Indeed, if $m=1$, then  $\sigma^n =\tau$ implies the existence of
 an  episturmian substitution   $\varrho$ such that $\tau=
\sigma \circ \varrho$, namely $\varrho=\sigma^{n-1}$.

We then remark that it is enough to consider the case when at
least one of $Z_{\sigma^n} $ or $Z_{\tau^m}$ has an error. Indeed,
otherwise, one has $Z_{\sigma^n}=[\sigma^n ]$ and
$Z_{\tau^m}=[\tau^m ]$. Since $[\sigma^n ]=[\tau^m ]$, this gives
$Z_{\sigma^n}= Z_{\tau^m}$. This implies that  $[\sigma]$ and
$[\tau]$ coincide up to the permutations  $\theta_{\sigma}$ and
$\theta_{\tau}$, since they are both prefixes of $Z_{\sigma^n}=
Z_{\tau^m}$. We then  can move $\theta_{\sigma}$
 inside $\mu(Z_{\tau^m})$  in such a way that
$\tau= \sigma \circ \varrho$.

So, we now assume that  $n \geq m \geq 2$  and at least one of
$Z_{\sigma^n} $ or $Z_{\tau^m}$  has an error. Here we provide the
list of lemmas handling all the cases needed for the proof of
Lemma \ref{lem:mainlem}. The precise statements and proofs are
provided in the next section. Indices $i$ and $i'$ below refer to Notation \ref{notation}.


\begin{itemize}
   \item
Lemma \ref{lem:2no}  states that if  $Z_{\sigma^n}$ or
$Z_{\tau^m}$ has an error,  then so does  the other one.
 \item
 Lemma \ref{lem:prefix} handles the case where
  the left propagated  errors of the first level start at the same index ($i=i'$).

\item Lemma \ref{lem:diff} handles the case  where  the  left propagated  errors of the first level errors   occur in
 $Z_{\sigma^n}$ and $Z_{\tau^m}$ with different letters ($a
\neq a'$).

\item
Lemma  \ref{lem:samelength}  handles the case where $\sigma$ and
$\tau$ have the  same  length.

\item
Lemma  \ref{lem:casepropa}  handles the case where  $Z_{\sigma^n}$
and $Z_{\tau^m}$  both have   left propagated  errors of the first level,  with the same letter,  with
both being of type I.
\item
Lemma \ref{lem:type33}  handles the case where where
$Z_{\sigma^n}$ and $Z_{\tau^m}$  both have   left propagated  errors of the first level, with the same
letter,  with  both being of type II.
\item
Lemma \ref{lem:type32}  handles the case where  $Z_{\sigma^n}$ and
$Z_{\tau^m}$  both have   left propagated  errors of the first level, with the same letter,  but with
different types.
\item
Lastly, Lemma \ref{lemma:oneletter} handles the case where  one of
the  substitutions  $\sigma$ and $\tau$ admits only one letter in
its normal form.
\end{itemize}

\subsection{Decomposition lemmas} \label{subsec:lemmas}

We now state and prove the lemmas that are used in the proof of
Lemma \ref{lem:mainlem}. All lemmas  stated in this section,
except Lemma \ref{lemma:oneletter}  involve substitutions  that
admit at least two distinct  letters in  their normalized
directive sequence. In Lemma \ref{lemma:oneletter} we handle the
one-letter case separately.
The statements and proof rely on  Notation \ref{notation}.

\begin{lemma}\label{lem:2no}  Let $\sigma$ and $\tau $ be episturmian  substitutions  such that
 $\sigma^n =\tau^m$, for $n,m \geq 2$, with   $\sigma$ and $\tau$  both admitting   at least  two  distinct letters
in their  normal form.
If
 one of the two words $Z_{\sigma^n}$ or $Z_{\tau^m}$    has an error,  then the  other one
also has   an error. \end{lemma}

\begin{proof}
We assume  w.l.o.g.  that  $Z_{\tau^m}$ has an error. We work by
contradiction  and we assume that $Z_{\sigma^n}$ has no error.
First notice that during the  normalization process the binary word
constituting of spins becomes lexicographically smaller with each
modification \eqref{eq:block_transf}: Indeed, each such
modification corresponds to replacing a factor $01\cdots 1$ by a
factor $10\cdots0$ of the same length. This means in particular
that if there is an error, then the word made of spins for the
normalized word is different  from the original word.

Now, due to $k'$-periodicity of the errors, $[\tau^{m}]^2$
contains an error (in fact, if $[\tau]^m$ contains a leftmost
simple error at position $s$, then $[\tau^{m}]^2$ contains a
simple error at position $((k'-1)m+s)$). Thus the binary spin
words are different for $[\tau^{m}]^2 $ and $ [\tau^{2m}]$.

However, $ [\sigma^{n}] =  [\tau^{m}]$ and   $ [\sigma^{2n}] =
[\tau^{2m}]$. 
Since  $Z_{\sigma^n}$ contains no error, we have
$[\sigma^{2n}]=[\sigma^{n}] ^2$, so, the words
  made of the spins are the same  for  $[\sigma^{2n}]$ and  for $[\sigma^{n}] ^2$, which yields  a contradiction with
   the  words   made of the spins  for  $[\tau^{m}]^2 $ and  $ [\tau^{2m}]$, which are not the same.
\end{proof}
\begin{lemma} \label{lem:prefix}
Let $\sigma$ and $\tau $ be episturmian  substitutions  such that
$\sigma^n =\tau^m$, with $ n \geq m\geq 2$, with  $\sigma$ and
$\tau$ both admitting   at least  two  distinct letters in their
normal form
  and with $Z_{\sigma^n}$ or $Z_{\tau^m}$  both having  an error. If
both    leftmost   errors  start at the same
index ($i=i')$, then there exists  an  episturmian substitution
$\varrho$ such that  $\tau =\sigma \circ \varrho$.
\end{lemma}

\begin{proof}
The normal form $[\sigma]$  deprived from its last  permutation letter $\theta_{\sigma}$ (i.e., the normalized   directive word $w_{\sigma}$),   is a prefix of  $Z_{\sigma^n}$ by definition.  It is also a prefix of  $ Z_{\tau^m}$. Indeed, we have the same spins in  $Z_{\sigma^n}$ and $ Z_{\tau^m}$  before the  beginning $i
$ of   the leftmost error.  By Lemma \ref{lem:propa}, between  $i$ and $k$, all the spins are  equal to $1$ in  $Z_{\sigma^n}$, and similarly  between  $i$ and $k'$ (we have $i'=i$), all the spins are equal to $1$ in  $Z_{\tau^m}$.
This implies that  $w_{\sigma}$  is a   prefix of $Z_\tau$.  The  desired  conclusion comes from   Remark \ref{rem:prefix}. \end{proof}

\begin{lemma} \label{lem:diff}  Let $\sigma$ and $\tau $ be episturmian  substitutions  such that $\sigma^n =\tau^m$, for
$ n \geq m \geq 2$,  with   $\sigma$ and $\tau$  both admitting at
least  two distinct letters in their normal form. We assume that $Z_{\sigma^n}$ and $Z_{\tau^m}$ both
contain  errors and that the leftmost errors have different
letters. Then there exists  an episturmian substitution $\varrho$
such that $\tau =\sigma \circ \varrho$.
\end{lemma}

\begin{proof}
Let $a$  be the letter of  the leftmost error in  $Z_{\sigma^n}$
for $\sigma$ and let  $a'$  be the letter  of the leftmost error
for $\tau$. We assume $a \neq a'$. We recall that  since $ n \geq
m$, then $k' \geq k$. Consider three cases according to the  types of
errors.
\begin{itemize}
\item
Assume first  that   the errors are of type I for both substitutions. One thus has $Z'_{\sigma^n}=[\sigma^n]= [\tau^m]=Z'_{\tau^m}$, by  Assertion (\ref{II10}) from
Lemma \ref{lem:propa}.  The  support $[i' ,k'+p']$  of the  leftmost propagated error of the first level in $Z_{\tau^m}$ cannot contain  $k+p$. Indeed,
 at position
$k+p$,   one has ${a}$  in   $Z'_{\tau^m}=[\tau^m]$ by  Assertion   (\ref{IIi}) ($a'\neq a$),  and  $\overline{a}$  in   $[\sigma^n]$   by  Assertion  (\ref{IIiii}).
Similarly,  the  support $[i , k+p]$  of the  leftmost propagated error of the first level in $Z_{\sigma^n}$ cannot contain  $k'+p'$.
Since $ k' \geq k  \geq i$, this   implies $k'+p '  >  i$, and thus  $k'+p'  > k+p$, and $i' >k+p$ from what precedes.
 Hence  the
 leftmost propagated error of the first level  in $Z_{\tau^m}$ starts strictly  after $k+p$. In particular   $k+p \leq  k'$ and $\pref_{k+p} Z_{\tau^m}=\pref_{k+p} [\tau]$.
Moreover,
the prefix of  length $k+p$ in   $Z_{\tau^m}$ is not modified by block-normalization  since $ k+p <i'$.
This yields  $\pref_{k+p} Z_{\tau^m}=  \pref_{k+p} [\tau^m] =\pref_{k+p} [\sigma^n]  \equiv   \pref_{k+p} Z_{\sigma^n}$.
By Assertion (\ref{II9}), there exists  an episturmian substitution  $\varrho_1$ such that
$ \mu( \pref_{k+p} [\sigma^n]) = \sigma \circ  \varrho_1$.
Since $ k +p \leq k'$,   there exists  an episturmian substitution $\varrho_2$ such that
$\tau =  \mu( \pref_{k+p} Z_{\tau^m})   \circ  \varrho_2$,
   which gives the existence of an episturmian substitution
 $\varrho$ such that  $\tau =\sigma \circ \varrho$.

\item The case where  the errors are of type II  for both substitutions is  impossible.  Let us prove it by contradiction and consider    the indices in  $Z'_{\sigma^n}$   and   $Z'_{\tau^m}$  where
the errors  end, namely    $k(n-1)+p   $ and $ k'(m-1)+p'$, respectively,  by   Lemma \ref{lem:type3}. Note that $a\neq a'$ implies that  $k(n-1)+p   \neq  k'(m-1)+p'$. Suppose that    $k(n-1)+p   > k'(m-1)+p'$.   By Assertion (\ref{III:viii}) from
 Lemma \ref{lem:type3},  the  letter at index  $k'(m-1)+p$
in  $[\tau^m]$ is $ {a'}$.  Since $a'\neq a$ and  since
$k'(m-1)+p$  is inside the error  for $\sigma$,  it occurs  with
spin $1$ in $[\sigma^n]$ by Assertion (\ref{III:vii}), and we  get
the desired contradiction  from $[\sigma^n]=[\tau^m]$. The same
reasoning applies if  the error in $\tau$ ends later than in
$\sigma$.


\item Now, suppose that   errors are  of different types.
Assume first  that $Z_{\sigma^n}$  admits    errors of
type I. We have $Z_{\sigma^n}[i]=\overline{a}$ and
$Z'_{\sigma^n}[i]=a$, as it is the first letter of the error, and
$Z_{\sigma^n}[k+p]=a$ and $Z'_{\sigma^n}[k+p]=\overline{a}$. By
Assertion (\ref{III:vii}) of Lemma \ref{lem:type3} and since $a\neq
a'$,  $k+p$ cannot belong to  the interval  $[i',p'+k'(m-1)]$  (which corresponds to the set of indices of letters  involved in the first and  second level of block-transformations in   $Z_{\tau}$), so either $k+p<i'$, or $k+p>k'(m-1)+i'$. In the first
case we have a contradiction with the same assertion at position
$i<k+p<i'$. In the second case we can take squares $\sigma^{2n}$
and $\tau^{2m}$ (or bigger powers), so that the inequality does
not hold with $2m$ instead of $m$. The case when $Z_{\tau^m}$
admits  errors of type I is symmetric. \qedhere
\end{itemize}
\end{proof}

\begin{lemma}\label{lem:samelength}
Let $\sigma$ and $\tau $ be episturmian  substitutions  such that $\sigma^n =\tau^m$,  for $ n,  m \geq 2$, with   $\sigma$ and $\tau$  both admitting   at least  two  distinct letters
in their  normal form.
 If  both  substitutions   have the same length, then there  exists   an  episturmian substitution
$\varrho$ such that  $\tau =\sigma \circ \varrho$  (and $\varrho$ is a  permutation).
\end{lemma}

\begin{proof} First note that  since $|\sigma|=|\tau|$, then
$n=m$.  We also note that in the case of equal
lengths, we must have, for each letter $\tilde{b}$ occurring in
$Z_{\sigma}$,  that it  also occurs in $Z_{\tau}$ at the same
positions, and moreover $\theta_{\sigma}(b)=\theta_{\tau}(b)$.
Otherwise, in the normal form we have different letters in
the second period at a same position in $[\sigma^n]$ and in
$[\tau^n]$.

By Lemma \ref{lem:2no}, $Z_{\sigma^n}$  and $Z_{\tau^n}$  both  have an error.
We distinguish  two cases according to the fact that they start at the same position, or not.

\begin{itemize}
\item Let us assume that   the leftmost  errors in $Z_{\sigma^n}$  and $Z_{\tau^n}$ start at the same position, i.e., $i=i'$.  Lemma \ref{lem:prefix}
together with   $|\sigma|=|\tau|$  imply that   there exists  a permutation
$\varrho$ such that  $\tau =\sigma \circ \varrho$.

\item Assume  now that $Z_{\sigma^n}$  and $Z_{\tau^n}$   are such that their leftmost errors   start at different
positions, i.e., $ i\neq i'$. We are going to prove that this case is impossible.


We assume $i'<i$ w.l.o.g. This implies that
$Z_{\sigma^n}$ and $Z_{\tau^n}$ coincide before $i'$ (i.e., in
$[1,i'-1]$), because nothing changes here during the
normalization. Let $c$ be the first non-barred letter. The letter
of the error in $Z_\tau$ is the first non-barred letter in the
second period, which is $\theta_{\tau}(c)$. The same argument
applies for $\sigma$, so the letters of the leftmost errors
coincide in both substitutions.

Now consider two further cases.
\begin{itemize}
\item  We assume $\theta_{\sigma}(a)=a$ (hence by the above $\theta_{\tau}
(a)=a$).

Since all the errors are
$a$-errors,  the respective total numbers of occurrences of $a$'s
and $\overline{a}$'s do not change during the normalization. Hence
they are   equal in both $Z_{\sigma^n}$ and $Z_{\tau^m}$, and
since $|\sigma|=|\tau|$, they are also equal in the first period,
i.e., $|\pref_{k}Z_{\sigma^n}|_a=|\pref_k Z_{\tau^n}|_a$. Note
that we use   the fact  that all  errors are $a$-errors in
$Z_{\tau^n}$. We recall that 
$Z_{\sigma^n}$ and $Z_{\tau^n}$ coincide before $i'$ (i.e., in
$[1,i'-1]$), 
and in $[i,k]$ (this is due to the structure of
errors:  all the letters there have spin $1$). In $Z_{\tau^n}$,
all the letters have spin $1$ also between $i'$ and $i$. After
normalization we have $[\tau^n][i']=a$, so we must have
$Z_{\sigma^n}[i']=a$ (since $i'<i$, i.e.,  in $Z_{\sigma^n}$, the
position $i'$ is before the occurrences of the errors, so it does
not change during normalization). So, we have $|\pref_{k}
Z_{\sigma^n}|_{a}>|\pref_{k} Z_{\tau^n}|_{a}$, a contradiction.

\item We now assume $\theta_{\sigma}(a)\neq a$ (and hence $\theta_{\tau}(a)\neq a$).
In particular, we have no type II error by Lemma \ref{lem:type3}
in $Z_{\sigma^n}$  and in $Z_{\tau^n}$, and the
letter of the error changes in every period. Consider the position
$i'$. As in the previous case, we have $Z_{\sigma^n}[i']=a$ (since
$i'<i$), and thus, in the last period, this gives
$Z_{\sigma^n}[k(n-1)+i'] = {\theta_{\sigma}^{n-1}(a)}$.    We have
$Z_{\tau^n}[k(n-1)+i'] =
[Z_{\tau^n}][k(n-1)+i']=\overline{{\theta_{\tau}^{n-1}(a)}}$
(since   $k(n-1) + i ' >  k(n-1) +p'$, we are after the  end of
the  last error in $Z_{\tau^n}$ and we use Assertion (\ref{III:ix})
from Lemma \ref{lem:type3}).  So, we must have
$[Z_{\sigma^n}][k(n-1)+i']=\overline{{\theta_{\sigma}^{n-1}(a)}}$,
while $Z_{\sigma^n}[k(n-1)+i'] = {\theta_{\sigma}^{n-1}(a)}$. This
implies that during the normalization for $\sigma^n$,  the letter
at index  $k(n-1)+i' $ had to change its spin. This can happen
only if this position is  inside the error between  the $(n-1)$-th
and  the $n$-th periods for $\sigma^n$,   and only if this letter
is the letter of the error. But the letter of the error is
$\theta_{\sigma}^{n-2}(a)$, and we are in the case
$\theta_{\sigma}(a)\neq a$, which yields  the desired
contradiction. \qedhere
\end{itemize}
\end{itemize}
\end{proof}

 In the proof of the next lemma (namely Lemma \ref{lem:casepropa})
  we will use several times the following  argument.
\begin{claim}\label{claim:alla}
Let $\sigma$ and $\tau $ be episturmian  substitutions  such that $\sigma^n =\tau^m$, for $ n > m\geq 2$, with   $\sigma$ and $\tau$  both admitting   at least  two  distinct  letters
in their  normal form. We assume  that  $Z_{\sigma^n}$ and $Z_{\tau^m}$ have  an error with the same letter $a$ for their leftmost error (with possibly  different types
of errors). Let the first propagated error of the first level  in $Z_{\tau^m}$ start
before the end of the first  period in $Z_{\sigma^n}$, i.e.,
$i'\leq k$. Then, either there is no central part  in $Z_{\tau^m}$,
or,  for any $i'\leq t \leq k$, the position $t$ is not in the
central part of the propagated error of the first level  in  $Z_{\tau^m}$. In
particular, there exists a letter ${b} \neq a$ (that occurs as
$\overline{b}$) in $[i',k']$ in  $Z_{\tau^m}$. \end{claim}

\begin{proof}  We assume that there exists  a   central part in  $Z_{\tau^m}$.  Assume  by contradiction that the position  $t$,
with $i'\leq t \leq k$, is in the  central part of
$Z_{\tau^m}$.  Clearly, we have $\overline{a}$'s at positions $t,\dots, k'$, since
we are in the central part of $Z_{\tau^m}$.  Consider the largest position $j< t$ with an
occurrence of $\tilde{b}$ for $b\neq a$  (such an occurrence exists  since $\sigma$   contains letters distinct from $a$ by assumption).

We  now look at  the last  periods  in $Z_{\sigma^n}$  and  $Z_{\tau^m}$,  by exploiting   $k$-periodicity and  $k'$-periodicity, respectively.
For  $Z_{\sigma^n}$ we get that
$Z_{\sigma^n}[k(n-1) +j,kn]=\tilde{{b'}}\tilde{{a'}} \cdots
\tilde{{a'}},$
where $b'\neq a'$; in fact,
$a'=\theta_{\sigma}^{n-1}(a)$ and $b'=\theta_{\sigma}^{n-1}(b)$.
On the other hand, $ Z_{\tau^m}[k'(m-1) +j,
k'm]=\tilde{{b''}}\tilde{{a''}} \cdots
\tilde{{a''}},$  where $b''\neq a''$; in fact,
$a''=\theta_{\tau}^{m-1}(a)$ and $b''=\theta_{\tau}^{m-1}(b)$.
Moreover, one has $a''=a'$ (by looking at the last letter).

Now consider the position $k(n-1) +j$ in $Z_{\tau^m}$. We have
$k(n-1) +j> k'(m-1) +j$, since $k<k'$ and $kn=k'm$. Therefore,
 we have $Z_{\tau^m}[k(n-1) +j]=\tilde{a''}=\tilde{a'}$, but
since the letters in $Z_{\tau^m}$ and $Z_{\sigma^n}$ at the same
positions coincide (up to spins), we have
 $Z_{\tau^m}[k(n-1) +j]=b'= \theta_{\sigma}^{n-1}(b)\neq a'=\theta_{\sigma}^{n-1}(a)$,  a
contradiction.

This  implies that  there cannot be  only $a$'s   in $[i' , k']$, otherwise  $i'$  would be in the central part.
\end{proof}

\begin{lemma}\label{lem:casepropa}
Let $\sigma$ and $\tau $ be episturmian  substitutions  such that $\sigma^n =\tau^m$, for $ n   > m \geq 2 $,  with   $\sigma$ and $\tau$  both admitting   at least  two  distinct letters
in their  normal form. We assume  that  $Z_{\sigma^n}$ and $Z_{\tau^m}$ have  an error with the same letter $a$ for their leftmost error and  that the
   $a$-error  is of type I   for both.   Then there exists an   episturmian substitution $\varrho$ such that
$\tau=\sigma\circ \varrho$.\end{lemma}

\begin{proof}
The  supports   of the propagated errors of the first level   in $Z_{\sigma ^n}$  are the intervals $[i+kj, p+k(j+1)]$, with $0 \leq j  \leq n-2$.

 If  $i=i'$,  the  desired conclusion comes from Lemma \ref{lem:prefix}. If $ k+p  < i'$, we denormalize $Z_{\tau}$
 to find the  normal form $[\sigma]$ as a prefix of some
word over $\A\cup \overline{\A}$ representing $\tau$, i.e., we
apply block-transformations $$xv \overline{x} \rightarrow
\overline{x} \overline{v} x    \quad (x \in {\mathcal A}, \  v \in
( {\mathcal A} \setminus \{x\})^*).$$ Indeed, the process of
normalization of the prefix of length $k+p$  for $Z_\sigma$  is
independent of the rest of the word.

Moreover one has  $ k +p \neq i'$ (hence $k+p > i'$) since otherwise, in the normalized form,  we would have $a$  at  $i'$  in $[\tau^m]$.
 On the other hand, it is  a final position for an error, hence   we would have $\overline{a}$  at  $k+p$  in $[\sigma^n]$.

 We thus assume $k+p > i'$   and $i \neq i'$ and we distinguish the  three following cases, namely  $i'<i$, $i<i' \leq k$,  and   $k<i'  <  k+p$.

\bigskip

{\bf Case 1.} We first  assume that $i'<i$. One has
$Z_{\tau^m}[i]= \overline{a}$, since  we have $i'<i  \leq  k'$
(indeed $i\leq  k < k'$) and due to Assertion (\ref{IIbis}) of
Lemma \ref{lem:propa}. Moreover,   $[\tau^m][i]= [\sigma^n][i]= a$
(by Assertion (\ref{IIii}) of Lemma  \ref{lem:propa},  since $i$
is the index of the  first letter  of  an error  in
$Z_{\sigma^n})$. Since  at index $i$  the letter $a$   has
different  spins in $Z_{\tau^m}$ and $[\tau^m]$ and  $i\in
[i'+1,k']$,  then  there is an  $a$-central part  $ \overline{a}
^+ \mid  {a }^+  $ in $Z_{\tau^m}$, and  the index $i$
 belongs to it by Assertion (\ref{IIiv}) of Lemma \ref{lem:propa}.
But,  due to Claim \ref{claim:alla} applied to $t=i$, this is impossible.  We thus reach a contradiction.

 \bigskip

{\bf Case 2.}  We  now  assume that $i<i'\leq k$.  At index $i'$,  we are in the   central part of
$Z_{\sigma^n}$. Indeed, we use an  argument    symmetric  to the one used in  the previous case:
$Z_{\sigma^n}(i')=\overline{a}$  and $[\sigma^n][i']= a$.  There thus exists a  central part in $Z_{\sigma^n}$. Let $s\geq 0$   be  such that the
central part  in  $Z_{\sigma^n}$ starts at $i'-s$; one has  $  i \leq i'-s \leq i'$.

The strategy  is  to use Proposition  \ref{prop:prefix}, Assertion \ref{prop:prefixi}.
To compare the numbers of
$a$'s in  prefixes of the same  length  of $Z_{\sigma^n}$ and $Z_{\tau^m}$, we  first split the prefixes of length $k$ into four parts, namely
$[1,i-1]$, $[i,i'-s-1]$, $[i'-s,i'-1]$,  and $[i',k]$. The numbers of $a$'s in the  first   part $[1,i-1]$ (before the  beginning of errors)  and in the  fourth parts $[i',k]$
 (made of $\overline{a}$'s since $i'$ is in the  central part of
$Z_{\sigma^n}$)  coincide.
In
the third part, we have
$Z_{\sigma^n}[i'-s,i'-1]=\overline{a}^{s}$,
$Z_{\tau^m}[i'-s,i'-1]=a^s$ (one has  $[\sigma^n][i'-s,i'-1]=  a^s$    by  Lemma  \ref{lem:propa} and this part is unchanged in the normalization
of $Z_{\tau^m}$).
Hence $| Z_{\sigma^n} [i'-s,i'] |_{a}= 0$ and  $|Z_{\tau^m}[i'-s,i'] |_{a}= s$.

For the  second  part, we   distinguish the  two cases $ i<i'-s $ and $i =i'-s$.

\begin{itemize}
\item We first  handle the case $i<i'-s$.  We will apply Proposition  \ref{prop:prefix}, Assertion (\ref{prop:prefixi}) to prefixes of length $k+s+1$.

  In
the second part $[i,i'-s-1]$  (which is outside of  the  central part),    the only index of occurrence of the letter $a$
where $Z_{\sigma^n}$ and $Z_{\tau^m}$ differ is $i$ by Assertion
(\ref{IIv}) of Lemma \ref{lem:propa}:
$Z_{\sigma^n}[i]=\overline{a}$ and $Z_{\tau^m}[i]=a$. By summing up, this gives   from above
$|\pref_{k } (Z_{\sigma^n})|_{a}-|\pref_{k } (Z_{\tau^m})|_{a}= -s-1$.

We now apply Assertion (\ref{prop:i}) of  Proposition \ref{prop:central}  to $\ell=i'$ and $t=i'-s$. Let $r$ stand for the number  of
$\overline{a}$'s  after $k+1$.  One has  $r\geq i- (i'-s) +1$, i.e., $r \geq s+1$.
This gives $s+1$ occurrences of $\overline{a}$'s in
$Z_{\tau^m}$ starting from index $k+1$.

Note that   $k' > k+s+1$. Otherwise, $ k'\leq k+s+1$ implies that
$k$ is in the central part of $Z_{\tau^m}$, yielding  a contradiction
 with  Claim \ref{claim:alla}.

 It remains to  compare the  parts
$[k+1,k+s+1]$.     We have
$Z_{\sigma^n}[k+1,k+s+1]={a}^{s+1}$,
$Z_{\tau^m}[k+1,k+s+1]=\overline{a}^{s+1}$.

 Summing up,  we have
equal numbers of $a$'s in the prefixes of length $k+s+1$, and
hence they are equal as substitutions by Assertion 1 of  Proposition
\ref{prop:prefix2}.

So, we can denormalize the prefix of length $k +s+1$ in $Z_{\tau
^m}$ and find $Z_\sigma$ as a prefix of a word over $\A\cup
\overline{\A}$ representing $\tau$.  Indeed, the substitution
which corresponds to  $\pref_{k +s+1}Z_{\tau ^m}$ equals  to the
substitution  which corresponds to  $\pref_{k +s+1}Z_{\sigma ^n}$.
There exists $h$ such that $ \mu(\pref_{k +s+1}Z_{\sigma ^n})=
\sigma \circ h= \sigma \circ \mu (Z_{\sigma ^n}[k+1\cdot\cdot
k+s])$. And  $ \mu(\pref_{k +s+1}Z_{\tau ^m}) $ is a prefix of
$[\tau]$ ($k'>k+s+1$). 

\item We  now assume $i=i'-s$.
One has from above   $|\pref_{k} (Z_{\sigma^n})|_{a}- |\pref_{k } (Z_{\sigma^n})|_{a}= -s$.

We now apply   Assertion
(\ref{prop:ii}) of Proposition \ref{prop:central}  to $l=i'$ and $t=i'-s$. Let $r$ stand for the number
of $\overline{a}$'s  after $k+1$.  One has  $r\geq i- (i'-s) $,
i.e., $r \geq s$. Hence  there are   at least $s$  occurrences of
the letter $a$  in $ Z_{\sigma^n}$ starting from index $k+1$. This
gives $s$ occurrences of $\overline{a}$'s in $Z_{\tau^m}$ starting
from index $k+1$ ($k'>k+s+1$ from above). Hence,  we have
$Z_{\sigma^n}[k+1,k+s]={a}^{s}$,
$Z_{\tau^m}[k+1,k+s]=\overline{a}^{s}$. Summing up,  we have equal
numbers of $a$'s in the prefixes of length $k+s$, and hence they
are equal as substitutions by Assertion (1) of  Proposition \ref{prop:prefix2},
\end{itemize}

\bigskip

{\bf Case 3.} We now  assume that  $k<i' <  k+p$.
We distinguish two cases, whether  $i$' belongs to  the central
part of $Z_{\sigma^n}$ or not.

\medskip

{\em Case 3.1.}  Let us  assume  that $i$' is not  in the central
part of $\sigma$. We  will show  that we can denormalize the
prefix of length $i'$ of $Z_{\tau^m}$
  to find $\sigma$ as a prefix of some
word over $\A\cup \overline{\A}$ representing $\tau$.  To do this,
we count the numbers of occurrences of $a$ and $\overline{a}$ in
the prefix of length $i'$, show that they are the same and apply
Proposition
  \ref{prop:prefix2}.

 Consider the largest  index $t$ of occurrence of a letter distinct from $a$
  before the position $i'$;  let $b$ stand for this letter. Since $k<i'$ by the conditions of Case 3
  and since $i'$ is not in the  central  part of $\sigma$, we have $k<t<i'$.

 Since $t$ is inside the error in $Z_{\sigma^n}$, we have
  $Z_{\sigma^n}(t)=\overline{b}$ and
  $[{\sigma^n}][t]={b}$ by Lemma
  \ref{lem:propa}, Assertion (\ref{IIi}).  Moreover, since $t<i'$, we  have    $[\tau^m][t]=Z_{\tau^m}[t]$, and thus $Z_{\tau^m}[t]={b}$ by $[\sigma^n]=[\tau^m]$.  So, applying Part 3 of Proposition \ref{prop:prefix2}, we get $|\pref_t(Z_{\sigma^n})|_{a}=|\pref_t(Z_{\tau^m})|_{a}-1$.

Now in the part $[t+1,i'-1]$ (which might as well beempty) we have
$a$'s in $Z_{\sigma^n}$ (since we are in the second part of
the error with  $k<t$), as well as in the normal form (as occurrences of
$\tilde{a}$ which are neither central, nor initial, nor final, according to
Lemma \ref{lem:propa}), so the same is true for  $Z_{\tau^m}$ (we
are before the beginning $i$ of the leftmost error, so everything there
coincides with the normal form).

Finally, since $k<i' <  k+p$ is in the second part of the error,
  by Assertion (\ref{IIbis}) of Lemma \ref{lem:propa}, we get $Z_{\sigma^n}[i']={a}$. We also have
$Z_{\tau^m}[i']=\overline{a}$ (as the first letter of an error by
Assertion (\ref{IIii})).

Summing up, we get equal numbers of occurrences of $a$'s in the
prefixes of length $i'$ of $Z_{\tau^m}$ and $Z_{\sigma^n}$. So, by
Proposition \ref{prop:prefix2}, Part 1, they are equal as
permutations, and  we can denormalize the prefix of length $i'$ of
$Z_{\tau}$ to find $Z_{\sigma}$ as a prefix of a word over $\A\cup
\overline{\A}$ representing $\tau$.
\medskip

{\em Case 3.2.} We now assume  that $i$' is   in the central  part of $Z_{\sigma^n}$.
 Let $s$   (with  $s\geq 1$ and $i\leq k-s$)  be the length of  the   longest suffix of $Z_\sigma$ filled with $\overline{a}$'s.
We
  have
  $k-s+1  \leq  k <i'< k+p$.

  One has
$Z_{\sigma^n}[k+1, i']= a^{i'-k}$,
 $ Z_{\tau^m}(i')=\overline{a} $, $ [{\tau^m}][i']=[\sigma^n][i']=a$, and $ Z_{\sigma^n}[i']=a$.
  At index $k$ in $Z_{\sigma^n}$, one has $\overline{a}$  (by Assertion (\ref{IIii}) of Lemma \ref{lem:propa}).
    In $Z_{\tau^m}$, from index $ k-s+1$ to $i'-1$, one has  only $a$'s; otherwise it would create  an error  and  we  also use that $i'$ was not the  beginning of the error.

  Similarly as in the previous case, to compare the numbers of
$a$'s in prefixes of the same length, we split the prefixes of length $i'$ into   four  parts, namely
$[1,i-1]$, $[i,k-s]$, $[k-s+1,k]$, $[k+1,i']$.
For the  first part $[1,i-1]$, one has the same number of $a$'s in  $Z_{\sigma^n}$ and $Z_{\tau^m}$ (by being before the errors).
For  $[k-s+1,k]$, one has  $|Z_{\sigma^n}[k-s+1,k]|_{a}-  |Z_{\tau^m}[k-s+1,k]|_{a}=-s$.
For $[k+1,i'-1]$,  we have only  $a$'s  for both; moreover,  $ Z_{\tau^m}(i')=\overline{a} $ and $ Z_{\sigma^n}[i']=a$.  This  yields   $|Z_{\sigma^n}[k+1,i']|_{a}-  |Z_{\tau^m}[k+1,i']|_{a}=+1$.

For the second part, we  distinguish the two cases  $i  <k-s$ and $i=k-s+1$.
Note that the  case  $i=k-s$ is  impossible. Indeed, if the error starts before the index  $k-s+1$  ($i<k-s+1$),  we have
$Z_{\sigma} [k-s] = \overline{b}$ for some $b \neq a$).

\begin{itemize}
\item  We assume $i  <k-s$.

In the second part  $[i,k-s]$,  the only index of occurrence
 of the letter $a$  where $Z_{\sigma^n}$ and $Z_{\tau^m}$ differ is $i$ with  $Z_{\sigma^n}(i)=\overline{a}$.
By summing up, this gives from above
$|\pref_{i'}(Z_{\sigma^n})|_{a}-  \pref_{i'}(Z_{\tau^m})|_{a}=-s$.

We apply Proposition \ref{prop:central} to $\ell=i'$ and $t=k-s+1$. We recall that  $r$ stands for the number of $a$'s  after $k+1$ in $Z_{\sigma^n}$. This gives
$r \geq i'- (k-s+1) +1$. Hence  one has at least  $s$   $a$'s after $i'+1$ (included).
This gives  $|\pref_{i'+s}(Z_{\sigma^n})|_{a}=  \pref_{i'+s}(Z_{\tau^m})|_{a}$.
 \item   We  now assume  $s= k-i+1$. Applying  Proposition \ref{prop:central}, one has at least  $s-1$   $a$'s after $i'+1$ (included).
 This gives  $|\pref_{i'+s-1}(Z_{\sigma^n})|_{a}=  |\pref_{i'+s-1}(Z_{\tau^m})|_{a}$. \qedhere
\end{itemize}
\end{proof}

We will make use of the following observation:

\begin{claim}\label{lem:type3a} Let $\sigma$ and $\tau $ be episturmian  substitutions  such that $\sigma^n =\tau^m$, with $ n \geq m \geq 2$,   with   $\sigma$ and $\tau$  both admitting   at least  two distinct   letters
in their  normal form.
 We assume  that  $Z_{\sigma^n}$ and $Z_{\tau^m}$ have  an error with the same letter $a$ for their leftmost error.
If  errors  are of  type II in  $Z_{\tau^m}$,  then
$\theta_{\sigma} (a)=a$. Similarly,  if   errors are of type II in  $Z_{\sigma^n}$,  then
$\theta_{\tau} (a)=a$.\end{claim}

\begin{proof}  We assume that  errors are of type II in $Z_{\tau^m}$. In particular,  this yields   $m \geq 3$ by Assertion (\ref{III:i}) of Lemma \ref{lem:type3pre}. So, one has $n \geq 3$. Assume that  $\theta_{\sigma} (a) \neq a$. Then  errors in $Z_{\sigma^n} $ are of type
I (by Assertion (\ref{III:i}) of Lemma \ref{lem:type3pre}). Now,
consider the index $2k+p$. It is the end of the error in the third
period in $Z_{\sigma^n}$, so one has the letter
$\theta_{\sigma}(a)$. It has spin $1$ after normalization in
$[\sigma^n]$  by Assertion (\ref{IIiii}) of Lemma \ref{lem:propa}.
Due to Assertion (\ref{III:vii}) of Lemma \ref{lem:type3}, this can
only happen if $2k+p \not \in [i', m(k'-1)+p']$.  If $2k+p  \geq
m(k'-1) +p'$, we then can take squares $Z_{\tau^{2m}}$ and
$Z_{\sigma^{2n}}$ to yield  a contradiction by  having $m$ large
enough. If $2k+p<i'$, consider the position $k+i$ (and thus $k+i <
i'$). We have $Z_{\sigma}[k+i]= \overline{\theta_{\sigma}(a)}$ (as
it is the first letter of a propagated error of the  first level). After
normalization, we have $ [\sigma^n][k+i] = \theta_{\sigma}(a)$.
And in the normalized form $ [\tau^m]$, for letters which are not
equal to $a$, occurrences with  spin 0 can only be inside the
error (due to Assertion (\ref{III:iv}) of Lemma \ref{lem:type3}). We
reach a contradiction in both cases.

We assume  now that the error is of type II in $Z_{\tau^m}$. By  considering   squares, i.e., $\sigma^{2n}$ and $\tau^{2m}$, we can assume $m \geq 3$
and the proof works as above.
\end{proof}

The following lemmas handle the case analogous to Lemma
\ref{lem:casepropa} for the case when both errors are of type II
(Lemma \ref{lem:type33}) or  not (Lemma
\ref{lem:type32}).

\begin{lemma}\label{lem:type33} Let $\sigma$ and $\tau $ be episturmian  substitutions  such that $\sigma^n =\tau^m$, with $ n \geq m \geq 2$,  with   $\sigma$ and $\tau$  both admitting   at least  two  distinct  letters
in their  normal form. We assume  that  $Z_{\sigma^n}$ and $Z_{\tau^m}$ have  an error with the same letter $a$ for their leftmost error, and that errors
 are of type II for both substitutions.  Then, there exists an  episturmian
substitution $\varrho$ such that $\tau=\sigma\circ
\varrho$.\end{lemma}

\begin{proof}
Since the errors are of type II, it means that for each letter
$b\neq a$,  with $ b \in {\mathcal A}$, each occurrence of $\tilde{b}$ in both $Z_{\tau^m}$ and
in $Z_{\sigma^n}$  has spin $1$ by  Assertion (\ref{III:iii}) of
Lemma \ref{lem:type3}.  After normalization, each such occurrence
inside  $[i, p+k(n-1)]$    goes from spin $1$ to spin $0$, and outside this  interval,
 it keeps spin $1$ by Lemma \ref{lem:type3}. The same applies to   $[i', p'+k'(m-1)]$. This means that
each occurrence of $\tilde{b}$ should be either inside    $[i, p+k(n-1)]$ and $[i', p'+k'(m-1)]$,
or outside both intervals. So, both errors start with the same block
of $\tilde{a}$'s.
Now notice that since we have only block-transformations of letter $a$, the proportion of $a$ and
$\overline{a}$ must be the same in both $[\sigma]$ and $[\tau]$,
hence $i<i'$ (by Assertion (\ref{III:iv}) of Lemma \ref{lem:type3}).
Now since $[\sigma^n][i']=[\tau^m][i']=a$, the spin of  the letter at index $i'$ must
change during the normalization of $Z_{\sigma}$, which means that
it is in the central part of $Z_{\sigma}$, i.e.,
$Z_{\sigma}[i', k]=\overline{a}\cdots \overline{a}$. Now the
indices in $[i',k]$ cannot belong to the central part of $Z_{\tau}$ (by
Claim \ref{claim:alla}), so there exists a letter $b\neq a$ and an
index $\ell$  of occurrence of  $\tilde{b}$, with  $k< \ell \leq k'$, which is in fact
$\overline{b}$ in both $Z_{\sigma}$ and $Z_{\tau}$. So, by Assertion (\ref{prop:prefixii})
 of Proposition \ref{prop:prefix},  the
substitutions corresponding to  the prefixes of $Z_{\sigma}$ and
$Z_{\tau}$ of length $\ell$ are equal. It follows that we can
denormalize the prefix of $Z_{\tau}$ of length $\ell$ (which is
also a prefix of $[\tau]$), introducing $\tau$ as a composition of
$\sigma$ and an episturmian substitution. \end{proof}

\begin{claim}\label{lem:claim2}
Let $\sigma$ and $\tau $ be episturmian  substitutions  such that
$\sigma^n =\tau^m$, for $ n > m \geq 2$,   with   $\sigma$ and
$\tau$  both admitting   at least  two distinct  letters in their  normal
form. We assume that all the errors in $\sigma^n$ and $\tau^m$ are
of letter $a$, and let $i$ and $j$ be integers with $ik \neq jk'$,
$i<n$ and $j<m$. Let $M=\min (ik, jk')$ and $N=\max (ik, jk')$.
Then there exists a letter $b \neq a$ that occurs in $[M,N]$, and
a letter $c \neq a$ (which can be the same as $b$)  that occurs in
$[M+1,N+1]$.
\end{claim}
\begin{proof}
The proof  is similar to the one of Claim  \ref{claim:alla}.
\end{proof}

\begin{lemma}\label{lem:type32}
Let $\sigma$ and $\tau $ be episturmian  substitutions  such that $\sigma^n =\tau^m$, for $ n > m \geq 2$,   with   $\sigma$ and $\tau$  both admitting   at least  two distinct letters
in their  normal form. We assume  that  $Z_{\sigma^n}$ and $Z_{\tau^m}$ have  an error with the same letter $a$ for their leftmost error,  with different  types of errors.
Then,   there exists an   episturmian substitution $\varrho$ such that  $\tau=\sigma\circ
\varrho$. \end{lemma}

\begin{proof}
First note that since the letters of the leftmost errors in
$Z_{\sigma}$ and $Z_{\tau}$ are equal to $a$,
 Claim \ref{lem:type3a} implies that
 $\theta_{\sigma} ( a)= \theta_{\tau} ( a)= a$.

We  now prove that   the case   where $Z_{\sigma^n}$ has  errors
of type I and $Z_{\tau^m}$ has  errors of type II cannot hold.
Since in $\tau$ we have  errors of type II, we have $m\geq 3$
(by Assertion (\ref{III:i}) of Lemma \ref{lem:type3pre}) and hence
$n\geq 4$  (since $n>m$). Now consider parts outside the  support
of the errors of type I in $Z_{\sigma^n}$ (i.e., at positions
$k\ell+p+1,  \ldots  ,k\ell+k+i-1$ for $\ell=1, \ldots, n-2$).
In particular, this means that $p\leq i+2$ (we have $p<i$ by
definition, and $p=i+1$ means an error of type II by Corollary
\ref{IIsuppbis}). Since errors are  of type I, each part
corresponding to a value of $\ell$ should contain $\tilde{b}$ for
some $b\neq a$, also by Corollary \ref{IIsuppbis} ($b$ could be
different in different parts due to the  permutation
$\theta_{\sigma}$). Let $k+t$ stand for the index of  such an
occurrence in the second period   (hence $p<t<i$), i.e.,
$Z_{\sigma^n}[k+t]=\tilde{b}.$ Since the spin is the same in all
the periods of $\sigma$, and at least one of them is covered by
the support $[i', (m-1)k'+p']$ of the error of type II in $Z_{\tau^m}$ ($m\geq 3$ and $n\geq 4$),
$[\sigma^n][k+t]=b$ (by Lemma \ref{lem:type3}, Assertion
(\ref{III:vii})).  This implies $Z_{\sigma^n}[k+t]=b$ (since the
spin is not changed after normalization outside the supports of the errors of type
I). By $k$-periodicity, this implies that in the first period in
$Z_{\sigma^n}$,  one has $Z_{\sigma^n}[t]=b$.  We also have $[
\sigma^n][t]=b$. Due to the structure of the error of type II,
all the occurrences of  the $ \tilde{b}$'s  have spin $1$ in
$Z_{\tau^m}$ (by Assertion (\ref{III:iii}) of Lemma
\ref{lem:type3}), and  in  $[\tau^m]$
 all the occurrences of the $b$'s are inside the error (by Assertion (\ref{III:vii}) of Lemma \ref{lem:type3}); outside nothing changes, so they keep their spin equal to $1$.
One thus has  $i'<t$  and $k(n-1)+t <  k'(m-1)+p'$,  which implies $t-k< p'-k'$ since $kn=k'm$.
However, $ t>i' $ together with   $ i'>p'$ implies   $t>p'$. The inequality  $ k-t > k'-p'$ implies  $ k>k'$,  which contradicts
$k'>k$.

We thus assume that  $Z_{\sigma^n}$ has  errors of type
II, and $Z_{\tau^m}$ has  errors of type I.  Symmetrically to the  above arguments,
we get that  each occurrence of $\tilde{b}$ at a position $t$ in
$Z_{\tau^m}$, with $k'l+p'<t<k'l+i'$, $l=0, \dots, m-1$, must be
an occurrence of $b$ and hence must be inside the type II error in
$Z_{\sigma^n}$. These considerations work for $m>2$, and to make
them work in the case $m=2$, it is enough to take squares of
$\sigma^n$ and $\tau^m$. For each such $t$ with $p'<t<i'$, one
has  $t> i$ (by considering the  first period) and $k'-t>k-p$ (by
considering the last period). But, contrary to the previous case,
this does not complete the proof. In particular, $t> i$ implies
 $i'>i$, as well as $k'>i+k-p$.

We  now distinguish two cases, namely   $i<i'\leq k$ and   $i'>k$.

\bigskip

{\bf  Case 1.}  We first assume that  $i<i'\leq k$.

 By   Claim \ref{claim:alla}, there exists an occurrence  $t$ in $[i'+1,k']$  of a letter   different from $a$, say $b$.
One has $Z_{\tau^m}(t)= \overline{b}= Z_{\sigma^n}(t)$. We then use Proposition  \ref{prop:prefix}  for $t-1$.

Let us prove  that  $t>k$. The position    $i$' belongs to  the
central part of  $Z_{\sigma^n}$. Indeed,
$Z_{\tau^m}(i')=\overline{a}$ and  $[Z_{\tau^m}][i']=a$ by
Assertion (\ref{IIii})  of Lemma \ref{lem:propa}. But
$Z_{\sigma^n}(i')=\overline{a}$ since   $i' \leq k$ together  with
Assertion  (\ref{IIbis}) of Lemma \ref{lem:propa}. The position
$i'$  belongs to the support of the  leftmost error of
$Z_{\sigma^n}$, so by Lemma \ref{lem:type3} there can be a  change
in the spin only if $i'$ in the central part of the error in
$Z_{\sigma^n}$, so we have
$Z_{\sigma^n}[i',k]=\overline{a}\cdots\overline{a}$ and
$Z_{\sigma^n}[k+1]=a$. This corresponds to $Z_{\tau^m}[i',
k+1]=\overline{a}\cdots\overline{a}$ since $k+1 \leq k'$ by
Assertion  (\ref{IIbis}) of Lemma \ref{lem:propa}. Hence $t>k$.
Now we use  Proposition \ref{prop:prefix} for $t-1$. This gives  $
\pref_{t-1} Z_{\sigma^n} \equiv  \pref_{t-1} Z_{\tau^m}$ and so we
can denormalize $ \pref_{t-1} Z_{\tau^m}$ to find $[\sigma]$ as a
prefix of a word over $\A\cup \overline{\A}$ representing $\tau$.

\bigskip

{\bf  Case 2.}  We  assume  now $i' > k$. We distinguish two
cases according to the occurrence or not  of  a letter $b \neq a$
in $[i'+1,k']$.
\begin{itemize}
\item[-]
Assume that  there exists an index $s$, with $i'<s\leq k'+1$, such
that $Z_{\tau^m}[s]=\tilde{b}$ for some $b\neq a$. Clearly, we
then have $Z_{\tau^m}[s]=\overline{b}$ by Assertion (\ref{IIi}) of
Lemma \ref{lem:propa}, and $Z_{\sigma^n}[s]=\overline{b}$ (by
Assertion (\ref{III:iii}) of  Lemma \ref{lem:type3}). By
Proposition \ref{prop:prefix}, the prefixes of $Z_{\tau^m}$ and
$Z_{\sigma^n}$ of length $s$ correspond to the same substitution
and also for $s-1$. Hence, $\pref_{s-1} Z_{\tau^m}$ can be
denormalized to find $Z_{\sigma}$ as a prefix of some word over
$\A\cup \overline{\A}$ representing $\tau$ since $ s
>k$ (see Remark \ref{rem:prefix}).

\item[-]
Assume  now that $Z_{\tau^m}[i' ,k'+1]=\overline{a}\cdots
\overline{a}$, $Z_{\tau^m}[k'+1]=a$. First, we deduce from Claim
\ref{lem:claim2} that there is an index $s$, with  $jk\leq s<i'$,
with  $j$  being  the largest integer such  that $jk  < k'$, such
that
 $Z_{\sigma^n}[s]=\overline{b}$ for some $b\neq a$. We use $M=jk$ and $N=k'$  with $[jk+1,i'-1]$.
Consider  indeed an occurrence of   a period of $Z_{\sigma^n}$
containing the position $i'$. It contains an occurrence of
$\tilde{b}$ for some $b\neq a$ and it is in fact an occurrence of $\overline{b}$
 (by Assertion (\ref{III:iii}) of Lemma
\ref{lem:type3}). Note also that we must have such an
occurrence before $i'$.

Now let us choose the largest such occurrence $ s$,  with  $ jk \leq s < i'$, such  that   $Z_{\sigma^n}[s]=\overline{b}$ for some $b\neq a$.
One has $Z_{\tau^m}[s]=b$. Indeed, since $s<i'$,  we have $Z_{\tau^m}[s] = [\tau^m][s]$. Moreover,  since  $s$ in the support  of  the error of    the second level of $Z_{\sigma^n}$ (by  Assertion (\ref{III:vii}) of  Lemma
\ref{lem:type3}),  we have
 $Z_{\sigma^n}[s]=\overline{b}$ (by Assertion (\ref{III:iii}) of  Lemma
\ref{lem:type3}). Now, by Assertion (\ref{prop:prefixiii}) of Proposition \ref{prop:prefix},  we have
$$|\pref_s(Z_{\sigma^n})|_{a}=|\pref_s(Z_{\tau^m})|_{a}+1 .$$
Observe  that  we  have
$Z_{\tau^m}[s, i']=b a^{i'-s-1} \overline{a}$ (the letters are $a$'s,   the  spins cannot alternate from $1$ to $0$,
 and   by definition of a   propagated error of the first level the spin  at index $i'-1$ equals $0$).

Let us assume that  $Z_{\sigma^n}[s+1]=a$. Then $Z_{\tau^m}[s,
i']=Z_{\tau}[s ,i']=b a^{i'-s-1} \overline{a}$, which we can
denormalize to $b \overline{a} a^{i'-s-1} (\overline{a})^{k'-i'}$.
Now we can apply Proposition \ref{prop:prefix} to $Z_{\sigma^n}$
and $b \overline{a} a^{i'-s-1}
(\overline{a})^{k'-i'}Z_{\tau^{m-1}}$ for the index $s+1$. With
the denormalization we did before, we get that we can denormalize
the prefix of length $i'$ of a word over $\A\cup \overline{\A}$
representing $\tau$ to find $\sigma$.

We  now  conclude the proof by  showing  that the case
$Z_{\sigma^n}(s+1)=\overline{a}$ is impossible. In this case   $s+1$
belongs to the end of some  copy of $Z_\sigma$ (in the beginning of the
propagated error of the first level), i.e., $ s+1 \in [ kj+i, k(j +1)].$   We recall that $s+1 \leq i'$ and $i'>k$.

First suppose that $ (j+1)k \neq k'$. By Claim \ref{lem:claim2},
we
 have an occurrence of $\overline{c}$ ($c \neq a$) at some position
$t$ between $i'$ and the end of the copy of $Z_{\sigma}$. Take the
smallest such $t$. Then,  the desired contradiction comes from
estimating the difference between the numbers of
occurrences of $a$ in $Z_{\sigma^n}[s+1 , t-1]$ and in
$Z_{\tau^m}[s+1,  t-1]$. On  the one hand, since
$Z_{\sigma^n}[s+1,t-1]=\overline {a}^{t-s-1}$,
we have $|Z_{\sigma^n}[s+1 , t-1]|_{a}=0$, and so
\begin{equation}
\label{eq:dif1}|Z_{\tau^m}[s+1 , t-1]|_{a}-|Z_{\sigma^n}[s+1 ,
t-1]|_{a}\geq 0.\end{equation}

On the other hand, we can estimate this difference  by considering
prefixes of length $s$ and $t$  together with  Proposition
\ref{prop:prefix}. Since $Z_{\sigma^n}[s]=\overline{b}$ and
$Z_{\tau^m}[s]=b$, by   Proposition \ref{prop:prefix},  Assertion \ref{prop:prefixiii}, one has
\begin{equation}
\label{eq:dif2}|\pref_s Z_{\sigma^n}|_a= |\pref_s
Z_{\tau^m}|_a-1.\end{equation}

We first assume  $Z_{\tau^m}[t]=\overline{c}$. Since
$Z_{\sigma^n}[t]=\overline{c}$,
by Proposition \ref{prop:prefix}, Assertion
\ref{prop:prefixii}, one has
\begin{equation}
\label{eq:dif3}|\pref_t Z_{\sigma^n}|_a= |\pref_t
Z_{\tau^m}|_a.\end{equation}
Subtracting \eqref{eq:dif2} from \eqref{eq:dif3}, we get a
contradiction with \eqref{eq:dif1}.

We now assume $Z_{\tau^m}[t]=c$. Then, by Proposition \ref{prop:prefix},
Assertion \ref{prop:prefixiii}, one has \begin{equation}
\label{eq:dif4}|\pref_t Z_{\sigma^n}|_a= |\pref_t
Z_{\tau^m}|_a-1.\end{equation} In addition, as  $Z_{\tau^m}[t]=c$,
the error in $Z_{\tau^m}$ ends before $t$, i.e., there exists an
index $j$,  with $i' < j < t$, such that $Z_{\tau^m}[j]=a$. This implies
$|Z_{\tau^m}[s+1,   t-1]|_{a}  \geq 1$, and thus we have an
inequality  even stronger than  \eqref{eq:dif1}:
\begin{equation} \label{eq:dif5}|Z_{\tau^m}[s+1 , t-1]|_{a}-|Z_{\sigma^n}[s+1 ,
t-1]|_{a}\geq 1.\end{equation} Subtracting \eqref{eq:dif2} from
\eqref{eq:dif4}, we get a contradiction with \eqref{eq:dif5}.

Now assume $ k' =  (j+1)k$. We show that this case is not possible.
We  consider  the position $k(n-1-j)+s$  which  corresponds to $s$   in the last period of $Z_{\sigma^n}$.
One has   $[\sigma^n] [k(n-1-j)+s]= Z_{\sigma^n} [k(n-1-j)+s] =\overline{b}$ since
$k(n-1-j)+s > k(n-1)+p$. One has   $[\tau^m] [k(n-1-j)+s]= Z_{\tau^m} [k(n-1-j)+s] ={b}$
since $b$ has spin $0$  and  does not change its spin during the normalization ($ b \neq a$). Hence, the desired contradiction. \qedhere
\end{itemize}
\end{proof}

We  now   consider the one-letter case.
\begin{lemma}\label{lemma:oneletter}
Let $\sigma$ and $\tau $ be episturmian substitutions such that
$\sigma^n =\tau^m$,  with $n\geq m \geq 2$. We assume that either
$\sigma$, or $\tau$, or both contain only one letter in its normal
form. Then  there  exists  an  episturmian substitution $\varrho$
such that  $\tau=\sigma\circ \varrho$.
\end{lemma}

\begin{proof}   We distinguish two cases according to the substitution whose normal form contains only one letter  (the case where both words contain
only one letter is handled inside Case 1).

\bigskip

{\bf Case 1.} Assume first that  $\sigma$  contains only one letter
in its normal form. We have
$[\sigma]=a^{s}\overline{a}^t\theta_{\sigma}$, with $s,t\geq 0$, $s+t\geq 1$.

\bigskip
We   now distinguish two further cases  according to the fact that
$\theta_{\sigma}(a)=a$ or not.

\begin{itemize}
\item Assume first  that  $\theta_{\sigma}(a)=a$.  One has $[\sigma^n]=a^{sn}\overline{a}^{tn}\theta_{\sigma}^n$. We have only letters $\tilde{a}$'s
in $[\sigma^n]$, so  in $[\tau^m]$ as well. It follows that
$\theta_{\tau}(a)=a$,
$[\tau]=a^{s'}\overline{a}^{t'}\theta_{\tau}$ for some integers
$s',t'\geq 0$, $s'+t'\geq 1$, and
$[\tau^m]=a^{s'm}\overline{a}^{t'm}\theta_{\tau}^m$. Hence
$s'm=sn$, $t'm=tn$, and since $n\geq m$, we have $s\leq s'$, and
$t'\leq t$. Applying several times the block-transformation
$a\overline{a}\to \overline{a}a$, we can denormalize $[\tau]$,
which gives $[\tau] \equiv a^s \overline{a}^t a^{s'-s}
\overline{a}^{t'-t} \theta_{\tau}$ and we conclude  as in Remark
\ref{rem:prefix} by inserting   $\theta_{\tau}^{-1}.$

\item  Now consider the case where $\theta_{\sigma}(a)\neq a$.
One has
$$[\sigma^n]=a^{s}\overline{a}^{t} \mid
(\theta_{\sigma}(a))^{s}\overline{\theta_{\sigma}(a)}^{t}\mid
(\theta^2_{\sigma}(a))^{s}\overline{\theta^2_{\sigma}(a)}^{t}\mid
\cdots
\mid(\theta^{n-1}_{\sigma}(a))^{s}\overline{\theta^{n-1}_{\sigma}(a)}^{t}\mid \theta_{\sigma}^n.$$

If $|\tau|=|\sigma|$ (and hence $m=n$),  inspection of  the first period of $[\sigma^n]= [\tau^m]$
yields that $[\tau]$ also
contains only $\tilde{a}$'s. Similarly, inspection of  the second period of $[\tau^m]$   yields $\theta_{\tau}(a)=\theta_{\sigma}(a)\neq a$, etc. So, $[\tau]^m$ is of the same form as $[\sigma]^n $ up to  the
permutations,
$[\tau]=a^s\overline{a}^t\theta_{\tau}$, and we have
$\tau=\sigma\circ \varrho$ with
$\varrho=\theta_{\sigma}^{-1}\theta_{\tau}$.

Assume now  $|\tau|\neq|\sigma|$. If  $\tau$   also  contains only one letter
in its normal form, we conclude as   above in the case  where $\theta_{\sigma}(a)=a$.   We thus assume now that  $[\tau]$ contains at least two
distinct letters. 
If  $Z_{\tau^m}$ contains
no  error, or  if  its  leftmost error starts at after the index $ k=|\sigma|$ (i.e., $i'>k$),
then  $w_{\sigma}$ is  a prefix of $w_{\tau}$, and  we conclude by Remark \ref{rem:prefix}.
So, it remains to consider the case where there is an
error in $Z_{\tau^m}$  which starts at $i' \leq k$;   the letter of this error is thus  $a$. Let us prove that we reach a contradiction.
The letter at the
ending position $k'+p'$ of the first propagated error of the first level (in $Z'_{\tau^m}$)  must also be $\tilde{a}$. 
Hence, the  leftmost error in $Z'_{\tau^m}$  (with support $[i',k'+p']$)  must contain all the $\tilde{b}$'s  (with $b=\theta_{\sigma}(a)$)    from the second period of $\sigma^n$ (i.e.,
with support $   [k+1,2k]$). After normalization,  the
letters $\tilde{b}$'s have spin $ 0$  (by Lemma \ref{lem:propa} and \ref{lem:type3}), since $ b \neq a$.  This implies  that $t=0$.
All letters  thus have spin $0$   in  $\tau^m]$; however, the last letter of a  leftmost error must have spin  $1$  in its normalized form    (at index $k'+p'$ by Lemma \ref{lem:propa}  for an error of type I  and  at index $k'(m-1)+p'$   for an error of type II by Lemma \ref{lem:type3}),  which yields the  desired  contradiction.
\end{itemize}

\bigskip

{\bf Case 2.} Now assume that  $\tau$ contains only one letter in its normal form.
Since $|\sigma | \leq | \tau|$,  $\sigma$ contains only one letter in its normal form.
Indeed, the condition $[\sigma^n]=[\tau^m]$
guarantees we have the same letters probably with different spins
at the same positions inside both $[\sigma]$ and $[\tau]$. The
case where $[\sigma]$ contains only letter  has  already
been handled in Case 1, which ends the proof.
\end{proof}

\section{Variations on the notion of  rigidity}\label{sec:morerigidity}




In the present paper, the  notion of  rigidity  has been  based so far   on  the study of the set of  substitutions having a  same  given   fixed point $u$. We have indeed focused
on $$\SStab{u} =\{  \tau  \mid    \tau \mbox{ is a  substitution  with
}   \tau(u)=u\}.$$  In this section,
we  first discuss  further  relevant variations    on the notion  of  rigidity inspired by \cite{Krieger:08,DiekertKrieger:09},
 and then provide an example illustrating  the notion of a weak rigidity.

We  recall that a  \emph{shift } $X$ is  a closed  shift-invariant set of infinite words of some ${\mathcal A}^{\mathbb N}$, where ${\mathcal A}$ is a finite alphabet.  Here the set ${\mathcal A}^{\mathbb N }$  is endowed with the product topology of the discrete topology on each copy of ${\mathcal A} $ and the {\em shift map}  $S$ is  defined by $S \left( (x_n)_{n \in \mathbb{N}} \right) = (x_{n+1})_{n \in \mathbb{N}}$.
The \emph{stabilizer}   of a  shift $X$   is then  defined   as the  monoid    
$$\SStab X =  \{ \tau \mid  \tau \mbox{  is a  substitution with }  \tau(X) \subset X \}.$$
In particular, for $u$ in $X$, one has $\SStab u  \subset \SStab  X $ when    $X$ is the orbit closure of $u$, i.e., $X$ is the closure of  $\{S^n u  \mid n \in {\mathbb N}\}$.

 We will
restrict to the {\em minimal case} for ease of simplicity,  that is, all
infinite words in $X$  have the same set of factors.
 Given a  primitive substitution $\sigma$,   the subshift $X_{\sigma}$ generated by $\sigma$  is   the set   of infinite words  having the same  language   as any  infinite  word fixed  by  some positive  power of $\sigma$ (they all have the same language by primitivity, hence  $X_{\sigma}$ is minimal).

We then  can  consider   variations of the notion of rigidity for a minimal subshift  $X_{\sigma}$, or else
for a primitive substitution $\sigma$ by considering the following
sets, called \emph{primitive stabilizers}\footnote{The inclusion is replaced   by an equality  by minimality.}:
$$\Stab{\sigma}= \{ \tau   \mid \tau \mbox{ is a primitive
substitution   with  } X_{\sigma}=X_{\tau}\},$$ $$\Stab{X} = \{
\tau \mid \tau \mbox{ is a primitive substitution   with  }
X_{\tau}=X\},$$
  and by  asking whether    their  primitive  stabilizer is
cyclic, i.e., it is   generated  by  a  single element.
In    \cite{Krieger:08}, the author  considers the notion of an  ``iterative stabiliser''  as the set of morphisms that generate, by iteration, a given shift  (together with the identity to make it a monoid).

We  also can define weaker  notions of rigidity. We recall that two
substitutions    $\sigma$ and $\tau$ over the same alphabet are
conjugate ($\sigma \sim \tau$) if there exists a word $w$ such
that
 $w \sigma(x)= \tau(x) w$ for   every letter  $x$.
A shift, an infinite word, or a substitution is  then  said {\em  weakly rigid} if for any two substitutions
$\sigma$, $\tau $ in  one of its stabilizers, $\sigma^k \sim \tau^m$ for some
positive integers $k, m$,  and for some given  equivalence relation.
Further   equivalence relations can also  lead   to  various notions of weak rigidity.  Consider  for instance  the operation $  w
\mapsto \overline{w}$ which  means that we replace  $ 0$ by $1$   and $1$
by $0$.
 We then can  define a notion of weak rigidity where
the equivalence  relation    relies on    the operation $  w
\mapsto \overline{w}$. Other equivalence relations  can be considered by   asking  whether  two substitutions
$\sigma, \tau$ have a common power (i.e.,  there exist $m,n$ such
that $\sigma^m= \tau^n$)   or whether  there are each power of a
common substitution (i.e., there exist $ k, \ell$ and $\varrho$
such that $\sigma=\varrho ^k$, $\tau= \varrho ^{\ell}$).

This can be illustrated via   the Sturmian case for two-letter
alphabets; see e.g.   \cite{BerFreSir}. Let $u$ be the Fibonacci
word  generated  by the Fibonacci substitution $\sigma \colon 0
\mapsto 01, 1 \mapsto 0$. One has $\Stab{u} \neq \Stab{\sigma}$.
Consider indeed   the  conjugate $\tau$  of the Fibonacci
substitution, i.e., $\tau \colon  0 \mapsto 10, \  1 \mapsto 0$. Both substitutions  $\sigma$ and $\tau$ are primitive.
One has  $\tau \not \in \Stab{u}$,  but  $\tau \in \Stab{\sigma}$,
$\Stab{u}=\{\sigma^n\}$  and $\Stab{\sigma}= \{ \mu \mid
\exists n, \    \mu \sim   \sigma ^n\}$.

We now detail  an example of   an infinite word $u$, obtained as  a fixed
point of a two-letter substitution $\sigma$,  which  displays  some rigidity
 involving  the equivalence relation   $  w
\mapsto \overline{w}$,  in the sense that  $\Stab{u} $ contains only
powers and products of the  two substitutions $\sigma $ and $\tau$, with     $\tau(0)=\overline{\sigma (1) }$  and $\sigma(1)=\overline{\tau (0) }$.

\begin{theorem} \label{theo:exnr}
Consider the following substitutions over the alphabet $\{0,1\}$
$$\sigma \colon 0 \mapsto  01,  \ 1 \mapsto  100110, \qquad
\tau  \colon  0 \mapsto  011001, \ 1 \mapsto  10.$$
Let $u$
be  the fixed point of $\sigma$ starting with $0$, i.e.,
$u=\sigma^{\infty}(0)$.
 Then, the  primitive  stabilizer $\Stab{u}$ of the word $u$  satisfies
$$\Stab{u}=\{\sigma^i, \tau^j, \sigma \tau^k, \tau \sigma ^{\ell}\mid i,j,k, \ell \in {\mathbb N} \}.$$ \end{theorem}

To prove Theorem  \ref{theo:exnr}, we  make use of several auxiliary
lemmas.

\begin{lemma}\label{lem:odd} In $u$, the factors $00$ and $11$ occur only at even indices.\end{lemma}

\begin{lemma} Each factor $v$ of $u$ of length at least 5 contains a factor $aa$, where $a$ is a letter.\end{lemma}

\begin{lemma} \label{lem:even}For each $\varphi \in \Stab{u}$, the lengths  $|\varphi(0)|$, $|\varphi(1)|$ are even.\end{lemma}

\begin{proof} Suppose  first  that  $|\varphi(0)|$ is odd, i.e.,  $|\varphi(0)|=2k+1$
for some integer $k$. Then
$$u=\varphi(0)\varphi(1)\varphi(1)\varphi(0)\cdots.$$

Let  $|\varphi(1)|=m$.
 If $|\varphi(0)|\geq 3$, then $\varphi(0)$
starts with $011$ and we find an occurrence of $11$ at position
$2k+1+2m+2$, which is odd; a contradiction with Lemma
\ref{lem:odd}. If $|\varphi(0)|< 3$, then $\varphi(0)=0$. Hence
$\varphi(1)$ must begin with $11$. We have
$$\varphi(0110)=0\varphi(1)\varphi(1)0,$$
which means that $\varphi(1)$ ends with $1$; otherwise we have an
occurrence of $00$ at index $2m+1$. But in this case we have an
occurrence of $111$ at index $m+1$,  which is not possible  since
$111$  cannot be  a factor of  $u$. We thus have proved that
 $|\varphi(0)|$ is even; let $p$ be such that  $|\varphi(0)|=2p$.

Now suppose $|\varphi(1)|$ is odd, i.e., $|\varphi(1)|=2\ell+1$
for some integer $\ell$.

\begin{itemize}
\item
If $|\varphi(1)|\geq 5$, then consider the prefix
$\varphi(0)\varphi(1)\varphi(1)$ of $u$. Since  $\varphi(1)$
contains a factor $aa$ and it is of odd length, it follows that
the prefix contains $aa$ at an odd position; a contradiction.
\item
If $|\varphi(1)|=3$, i.e., $\varphi(1) = abc$ for some $a,b,c
\in\{0,1\}$, then consider the prefix $\varphi(0)abcabc$ of $u$.
Considering factors at odd positions $2p+1$, $2p+3$ and $2p+5$, we
get that $a\neq b$, $c\neq a$, $b\neq c$, which is not possible
over a two-letter alphabet.

\item
If $|\varphi(1)|=1$, i.e., $\varphi(1) = a$ for some $a
\in\{0,1\}$, then consider the prefix $\varphi(0)aa$ of $u$. The
factor $aa$ of $u$ is at odd position $2p+1$, which is not
possible by Lemma \ref{lem:odd}. \qedhere
\end{itemize}\end{proof}

 Let  $u'$
 be the word over $\{A,B\}$  defined  as $u=\psi(u')$, where
$$\psi: A \mapsto 01, B \mapsto 10.$$
(Since this  does not cause  any confusion, we   use the same
notation for the  letters $A,B$  and  for their   ``value'' in
$\{0,1\}$, i.e., $A=01, B=10$. The same holds for the letters $C$
and $D$ defined below.)
Let $C=0110$, $D=1001$. It is not hard to see that $u$ is a
concatenation of   blocks  $C$ and $D$. Indeed, by Lemma
\ref{lem:odd}, $u$ is a concatenation  of blocks $01$ and $10$,
and since $\sigma(u)=u$, we have that $\sigma(u)$ is a
concatenation  of blocks   $\sigma(0)\sigma(1)=CC=01100110$ and
$\sigma(1)\sigma(0)=DD=10011001$.   This also implies that $u'$
consists of blocks $ABAB$ and $BABA$. We  now define the following
recurrence relations:
$$\begin{cases} T_0=C \\
T_n=T_{n-1}T_{n-1}\overline{T_{n-1}}\overline{T_{n-1}},\end{cases}$$
where $\overline{C}=D$, $\overline{D}=C$.  We recall that  the operation $  v \mapsto \overline{v}$ means that we replace  $ 0$ by $1$   and $1$ by $0$.
We then define the word
$u''$ over $\{C,D\}$  as $u''=\lim_{n\to\infty}T_n$.

\begin{lemma} \label{lem:varrho}
One has $u=\varrho(u'')$ where $\varrho$ is the morphism from $\{C,D\}^*$ to $\{0,1\}^*$ defined  by  $\varrho: C\mapsto 0110, D \mapsto 1001$.\end{lemma}

\begin{proof} The proof works by induction.
One has
\begin{align*}
\sigma(01)&=\varrho(T_0T_0),\\
 \sigma(10)&=\varrho(\overline{T_0}\overline{T_0}),\\
\sigma^{n+2}(01)&=\sigma^{n+1}(01100110)\\
&=\sigma^{n+1}(01)\sigma^{n+1}(10)\sigma^{n+1}(01)\sigma^{n+1}(10)\\
&=\varrho(T_nT_n)\varrho(\overline{T_n}\overline{T_n})\varrho(T_nT_n)\varrho(\overline{T_n}\overline{T_n})\\
&=\varrho(T_{n+1}T_{n+1}).
\end{align*}

The same  holds similarly for $\sigma(10)$. So, we proved that
arbitrarily long prefixes of $u$ and $\varrho(u'')$ coincide,
which completes the proof.\end{proof}

\begin{claim} \label{claim:CD} One has  $u''=\mu^{\infty}(C)$, where  $\mu$  is the substitution over $\{C,D\}$  defined as $\mu: C\mapsto CCDD,
D\mapsto DDCC$.
\end{claim}

\begin{proof} We  prove by induction on $n$ that
$T_n=\mu^n(C)$, $\overline{T_n}=\mu^n(D)$.
Indeed, $T_1=CCDD=\mu(C)$, $\overline{T_1}=DDCC=\mu(D)$.
Now
$$T_{n+1}=T_nT_n\overline{T_n}\overline{T_n}=\mu^n(C)\mu^n(C)\mu^n(D)\mu^n(D)=\mu^n(CCDD)=\mu^{n+1}(C).$$

The case of  $T_{n+1}$ is handled similarly. \end{proof}

We use the following notation  when ``desubstituting''. For $w$  word over $\{C,D\}$,  $\varrho^{-1}(w)$  stands for  the  word   over $\{0,1\}$
that satisfies  $w= \varrho (\varrho^{-1}(w))$. The same holds for  $ \psi^{-1}$, by noticing that the  desubstitution is unique in both cases.

\begin{lemma}If  $u$  has a  prefix of the form $vv$, then  there exists $i$ such that $|vv|=2\cdot 4^i$.\end{lemma} \label{lem:square}

\begin{proof} Suppose first that $|v|$ is divisible by $4$. Then the prefix  $vv$ corresponds
to a square $v'v'=\mu^{-1}(v)\mu^{-1}(v)$ in $u''$. If $|v'|$ is
 divisible by 4, we  continue to desubstitute and find a smaller square
 $\mu^{-1}(v')\mu^{-1}(v')$; we continue desubstituting
 until we find a square $v_0v_0$ of length not divisible by 4. If
 $|v_0|=1$, the lemma is proved; otherwise consider the following
 cases on the length of $v_0$ $\mod 4$, for which we  reach a contradiction. This will prove  that $|v_0|=1$.

\begin{itemize}
\item Assume that  $|v_0|=4k+1$ ($k>0$) or  $|v_0|=4k+3$. If $v_0$ ends with
$C$, then it ends with either $DCCC$ or with $DC$ (from the
structure of the word given by Claim \ref{claim:CD}). The second copy
of $v_0$ starts with $CCD$, therefore we have an occurrence of
$DC^3D$ or $DC^5D$ in $u''$, which is not possible. If $v_0$ ends
with $D$, the proof is symmetric.

\item Assume that  $|v_0|=4k+2$. For $k=0$ and $k=1$ we clearly do not have
squares. For $k>1$, the first occurrence of $v_0$ contains $D^4$,
and in $u''$ the factor $D^4$ can occur only at positions of the
form $8m-1$. Then in the second copy of $v_0$ the copy of the
factor $D^4$ occurs at positions  of the form  $8\ell+1$ or
$8\ell+5$, which is not possible in $u''$.
\end{itemize}

Now consider the case where $v$ is not divisible by $4$.
\begin{itemize}
\item
 If  $|v|=4k+1$ ($k>0$) or   $|v_0|=4k+3$, then $v$ contains an
 occurrence of $11$ or $00$;  hence its second copy occurs at an odd
 position, which is not possible by Lemma \ref{lem:odd}.
\item
 If  $|v|=4k+2$, then  the contradiction comes  from considering an    occurrence of $BB$ in
 $u'$: it can only occur at an even position, and the second copy of $v$  provides an  occurrence at an odd position. \qedhere
 \end{itemize}
\end{proof}

\begin{corollary}\label{cor:varphi01}  Let  $\varphi\in \Stab{u}$. We have $\varphi(01)=\varrho(T_nT_n)$ for some $n$, and $\varphi(10)=\varrho(\overline{T_n}\overline{T_n})$.\end{corollary}

\begin{proof} The prefix $01100110$ is a square prefix of $u$, so
$\varphi(01100110)$ is also  a square prefix of $u$,   and  its
length   is thus of the form $2\cdot 4^i$ by Lemma
\ref{lem:square}. This means that $\varphi(01)$ has length $2\cdot
4^{i-1}$, which in turn means exactly $\varphi(01)=\varrho(T_nT_n)$
for $n=i-2$.
\end{proof}
\begin{lemma}\label{lem:exlast}
For each $\varphi\in \Stab{u}$, there exist $\varphi'$ and $\varphi ''$  such that  one has either  $\varphi=\sigma \circ \varphi'$, or $\varphi=\tau \circ \varphi''$. \end{lemma}

\begin{proof} One  has $\varphi(u)=u$.  Suppose that $|\varphi(0)|<|\varphi(1)|$  (the other  case will be handled  analogously $ |\varphi(0)| \geq |\varphi(1)|$).
The word $u'$ (with $u =\psi (u'))$) consists of blocks $ABAB$ and
$BABA$. Since $|\varphi(0)|$ and $|\varphi(1)|$ are even (by Lemma
\ref{lem:even}), we have that $\psi^{-1}(\varphi(0))$ consists of
$A$'s and $B$'s, as well as $\psi^{-1}(\varphi(1))$. Since each
block $ABAB$ and $BABA$ is a concatenation of $A$ and $BAB$, we
get a natural factorization in these blocks. Now if this
factorization is a refining of the one with
$\psi^{-1}(\varphi(0))$ and $\psi^{-1}(\varphi(1))$, then the
lemma is proved.

If not, then we have some position of factorization of the form
$4k+2$, 3 or 4, and in addition $\psi^{-1}(\varphi(1))$ contains
$BB$, since the length of $\varphi(1)$ is bigger than the length
of $\varphi(0)$ (all the short substitutions can be checked with
an exhaustive search).

Now consider the prefix $\psi^{-1}(\varphi(0)\varphi(1))$ of $u'$,
whose  length  is  of the form $2\cdot 4^n$ by Corollary
\ref{cor:varphi01}. If the length of $\psi^{-1}(\varphi(0))$ is
divisible by 4, then all the positions of $\psi^{-1}(\varphi(0))$
and $\psi^{-1}(\varphi(1))$ are equivalent to $1$ modulo $4$. So,
the length of $\psi^{-1}(\varphi(0))$ is of the form $4\ell+1$, 2
or 3 for some $\ell$, and so is $\psi^{-1}(\varphi(1))$. Consider
the factor $\psi^{-1}(\varphi(1)\varphi(1))$ of $u'$. Consider an
occurrence of $BB$ in the first copy of $\psi^{-1}(\varphi(1))$,
which can only be at a position of the form $4m$ in $u'$. Then the
occurrence of $BB$ in the second copy cannot be of the form $4m'$.
\end{proof}

The proof of Theorem \ref{theo:exnr}  follows from  Lemma
\ref{lem:exlast} with the observation that $\sigma(01)=\tau(01)$,
$\sigma(10)=\tau(10)$, and hence $\sigma\circ  \sigma = \sigma
\circ \tau$, $\tau\circ \sigma = \tau\circ  \tau$ (by checking  it directly on letters).

\bibliographystyle{amsalpha}
\bibliography{Rigidbib}
\end{document}